\documentclass[11pt,reqno]{amsart}
\usepackage{a4wide}
\usepackage[english]{babel}
\usepackage{enumitem}
\usepackage{amssymb,amsmath,amsthm,amsfonts,latexsym,mathtools}
\usepackage{bbm}
\usepackage{appendix}
\usepackage{hyperref}
\usepackage{color}
\DeclareMathOperator{\tr}{Tr}
\makeatletter

\numberwithin{equation}{section}
\numberwithin{figure}{section}
\theoremstyle{plain}
\newtheorem{theorem}{\protect\theoremname}
\theoremstyle{remark}
\newtheorem*{remark}{Remark}
\theoremstyle{plain}
\newtheorem{lemma}[theorem]{\protect\lemmaname}
\newlist{casenv}{enumerate}{4}
\setlist[casenv]{leftmargin=*,align=left,widest={iiii}}
\setlist[casenv,1]{label={{\itshape\ \casename} \arabic*.},ref=\arabic*}
\setlist[casenv,2]{label={{\itshape\ \casename} \roman*.},ref=\roman*}
\setlist[casenv,3]{label={{\itshape\ \casename\ \alph*.}},ref=\alph*}
\setlist[casenv,4]{label={{\itshape\ \casename} \arabic*.},ref=\arabic*}

\makeatother

\providecommand{\lemmaname}{Lemma}

\providecommand{\casename}{Case}
\providecommand{\theoremname}{Theorem}

\allowdisplaybreaks
\begin{document}
	
	\title{Blow-Up Profile of 2D Focusing Mixture Bose Gases}
	
	\author{Dinh-Thi NGUYEN}
	\address{Dinh-Thi Nguyen, Mathematisches Institut, Ludwig--Maximilians--Universit\"at M\"unchen (LMU), Theresienstrasse 39, 80333 Munich, Germany, and Munich Center for Quantum Science and Technology (MCQST), Schellingstrasse 4, 80799 Munich, Germany.} 
	\email{\href{mailto:nguyen@math.lmu.de}{nguyen@math.lmu.de}}
	
	
	\maketitle

\begin{abstract}
	We study the collapse of a many-body system which is used to model two-component Bose--Einstein condensates with attractive intra-species interactions and either attractive or repulsive inter-species interactions. Such a system consists a mixture of two different species for $N$ identical bosons in $\mathbb R^2$, interacting with potentials rescaled in the mean-field manner $-N^{2\beta-1}w^{(\sigma)}(N^{\beta}x)$ with $\int_{\mathbb R^{2}}w^{(\sigma)}(x){\rm d}x=1$. Assuming that $0<\beta<1/2$, we first show that the leading order of the quantum energy is captured correctly by the Gross--Pitaevskii energy. Secondly, we investigate the blow-up behavior of the quantum energy as well as the ground states when $N\to\infty$ and either the total interaction strength of intra-species and inter-species or the strengths of intra-species interactions of each component approaches sufficiently slowly a critical value, which is the critical strength for the focusing Gross--Pitaevskii functional. We prove that the many-body ground states fully condensate on the (unique) Gagliardo--Nirenberg solution.
\end{abstract}

\tableofcontents

\section{Introduction}
After the first realization of Bose--Einstein condensate (BEC) in the laboratory in 1995 \cite{AndEnsMatWieCor-95,CorWie-02,DavMewAndDruDurKurKet-95,Ketterle-02}, theoretical studies have been developed for the one-component BEC. In that case, the NLS energy functional is commonly used to predict a collapse of the system when the interaction is attractive and the number of particles excesses a critical value \cite{BayPet-96,DalStr-96,UedLeg-98,MueBay-00,BaoCai-13,GuoSei-14}. This effect has been observed in some experiments \cite{BraSacTolHul-95,KagMurShl-98,SacStoHul-98,GerStrProHul-00,DonClaCorRobCorWie-01}. BEC with multiple species can display some interesting phenomena absent in single-component BEC. The simplest multi-component BEC is the binary mixture.

In this paper, we establish some results about 2D focusing mixture condensate in the critical regime of collapse. To be precise, we consider a Bose gas trapped into a quasi-2D layer by means of trapping potentials and we look at a nonlinear Schrodinger many-body system arising in a two-component BEC with attractive intra-species interactions and either attractive or repulsive inter-species interactions. It is described by the Hamiltonian for $N_{1}$ and $N_{2}$ identical bosons of different types in $\mathbb{R}^{2}$
\begin{align}
H_{N} = &  \sum_{i=1}^{N_{1}}\big(-\Delta_{x_{i}} + V_{1}(x_{i})\big)-\frac{1}{N_{1}-1}\sum_{1 \leq i < j \leq N_{1}}w_{N}^{(1)}(x_{i}-x_{j}) \nonumber\\
& + \sum_{r=1}^{N_{2}}\big(-\Delta_{y_{r}}+V_{2}(y_{r})\big)-\frac{1}{N_{2}-1}\sum_{1 \leq r < s \leq N_{2}}w_{N}^{(2)}(y_{r}-y_{s}) \nonumber\\
& - \frac{1}{N}\sum_{i=1}^{N_{1}}\sum_{r=1}^{N_{2}}w_{N}^{(12)}(x_{i}-y_{r}), \label{hamiltonian}
\end{align}
where $N = N_{1} + N_{2}$ is the total number of particles. The Hamiltonian \eqref{hamiltonian} acts on the Hilbert space
\begin{equation}\label{space:hilbert}
\mathcal{H}_{N} = \mathcal{H}_{N_{1}} \otimes \mathcal{H}_{N_{2}} := L_{\rm sym}^{2}(\mathbb{R}^{2N_{1}},{\rm d}x_{1},\ldots,{\rm d}x_{N_{1}})\otimes L_{\rm sym}^{2}(\mathbb{R}^{2N_{2}},{\rm d}y_{1},\ldots,{\rm d}y_{N_{2}}).
\end{equation}
Here, we denoted by $L_{\rm sym}^{2}(\mathbb{R}^{2N_j})$ the space of square-integrable functions in $(\mathbb R^{2})^{N_j}$ which are symmetric under permutations of the $N_j$ variables. The exchange symmetry is not present among variables of different types. The potentials $V_{1}$ and $V_{2}$, which can be different, stand for trapping potentials for each species, i.e, $V_{i}\geq 0$ and $V_{i}(x) \to +\infty$ as $|x| \to \infty$, for $i\in\{1,2\}$. The two-body interactions among particles of the same species and of different species approach a Dirac delta and are chosen in the form
\begin{equation}\label{assumption:potential-1}
w_{N}^{(\sigma)}(x) = -a_{\sigma}N^{2\beta}w^{(\sigma)}(N^\beta x),\quad  \sigma\in\{1,2,12\},
\end{equation}
for a fixed parameter $0<\beta<1$, and fixed functions $w^{(\sigma)}$  satisfy
\begin{equation}\label{assumption:potential-2}
w^{(\sigma)}(x)=w^{(\sigma)}(-x)\geq 0,\quad  (1+|x|)w^{(\sigma)},\widehat{w^{(\sigma)}}\in L^1(\mathbb R^{2}) \quad \text{and} \quad \int_{\mathbb R^{2}}w^{(\sigma)}(x){\rm d}x=1.
\end{equation}
The parameters $a_{1}>0$ and $a_{2}>0$, which are of order $1$, describe the strength of attractive intra-species inside each component. The inter-species interactions between two components of the system can be attractive ($a_{12}>0$) or repulsive ($a_{12}<0$). The choice of coupling constants proportional to $1/(N_j - 1)$ and $1/N$ ensures that the kinetic and the potential energies are comparable in the limit $N \to \infty$. In this limit regime, we assume that
\begin{equation} \label{number-particle}
\lim_{N\to\infty}\frac{N_{1}}{N}=c_{1}\in (0,1) \quad \text{and} \quad \lim_{N\to\infty}\frac{N_{2}}{N}=c_{2}=1-c_{1}.
\end{equation}
It is not restrictive to assume that the ratios $N_{1}/N$ and $N_{2}/N$ themselves are fixed, and so shall we henceforth.

We are interested in the large-$N$ behavior of the ground state energy per particle of $H_{N}$ in \eqref{hamiltonian} given by
\begin{equation}\label{energy:quantum}
E_{N}^{\rm Q} := N^{-1}\inf\sigma_{\mathcal{H}_{N}}H_{N} = N^{-1}\inf_{\substack{\Psi_{N}\in\mathcal{H}_{N} \\ \|\Psi_{N}\|_{L^{2}}=1}} \langle \Psi_{N}|H_{N}|\Psi_{N} \rangle
\end{equation}
and the corresponding ground states. Note that, for a fixed $N$, the potentials in \eqref{hamiltonian} are bounded from below, by the assumption \eqref{assumption:potential-2}. Therefore, the existence of a many-body ground state follows from the standard direct method in the calculus of variations. Furthermore, the energy per particle of the fully condensed trial function $u_{1}^{\otimes N_{1}}\otimes u_{2}^{\otimes N_{2}}$ is given by the $N$-dependent Hartree energy functional
\begin{align}
\mathcal{E}_{N}^{\rm H} (u_{1},u_{2}) & = c_{1}\int_{\mathbb R^{2}}\left[|\nabla u_{1}(x)|^{2}+V_{1}(x)|u_{1}(x)|^{2}-\frac{a_{1}}{2}|u_{1}(x)|^{2}(w_{N}^{(1)}\star|u_{1}|^{2})(x)\right]{\rm d}x \nonumber \\
& \quad  + c_{2}\int_{\mathbb R^{2}}\left[|\nabla u_{2}(x)|^{2}+V_{2}(x)|u_{2}(x)|^{2}-\frac{a_{2}}{2}|u_{2}(x)|^{2}(w_{N}^{(2)}\star|u_{2}|^{2})(x)\right]{\rm d}x \nonumber \\
& \quad  - c_{1}c_{2}a_{12}\int_{\mathbb R^{2}}|u_{1}(x)|^{2}(w_{N}^{(12)}\star|u_{2}|^{2})(x){\rm d}x, \label{functional:hartree-tow-component}
\end{align}
where $c_{1}$ and $c_{2}$ are the ratios defined in \eqref{number-particle}. It turns out that the leading order of the quantum energy is captured by the effective Hartree energy in the mean-field regime. In fact, the Hartree energy, which is obtained by taking the infimum of the Hartree energy functional in \eqref{functional:hartree-tow-component} under the constrain $(u_{1},u_{2})\in H^{1}(\mathbb R^{2}) \times H^{1}(\mathbb R^{2})$ and $\|u_{1}\|_{L^{2}} = 1 = \|u_{2}\|_{L^{2}}$, is an upper bound to the quantum energy
\begin{equation}\label{energy:hartree-two-component}
E_{N}^{\rm Q} \leq \inf_{\substack{u_{1},u_{2}\in H^{1}(\mathbb R^{2})\\ \|u_{1}\|_{L^{2}} = 1 = \|u_{2}\|_{L^{2}}}}\mathcal{E}_{N}^{\rm H}(u_{1},u_{2}) =: E_{N}^{\rm H}.
\end{equation}
The existence of an $N$-dependent Hartree ground state is again a consequence of assumption \eqref{assumption:potential-2}. When $N\to\infty$, since $w_{N}^{(\sigma)}\approx \delta_0$, for $\sigma\in\{1,2,12\}$, the Hartree energy functional formally boils down to the nonlinear Schr\"{o}dinger (NLS) energy functional
\begin{align}
\mathcal{E}^{\rm NLS}(u_{1},u_{2}) & = c_{1}\int_{\mathbb R^{2}}\left[|\nabla u_{1}(x)|^{2}+V_{1}(x)|u_{1}(x)|^{2}-\frac{a_{1}}{2}|u_{1}(x)|^{4}\right]{\rm d}x \nonumber \\
& \quad  + c_{2}\int_{\mathbb R^{2}}\left[|\nabla u_{2}(x)|^{2}+V_{2}(x)|u_{2}(x)|^{2}-\frac{a_{2}}{2}|u_{2}(x)|^{4}\right]{\rm d}x \nonumber \\
& \quad  - c_{1}c_{2}a_{12}\int_{\mathbb R^{2}}|u_{1}(x)|^{2}|u_{2}(x)|^{2}{\rm d}x.\label{functional:nls-two-component}
\end{align}
We can therefore expect that the NLS energy $E^{\rm NLS}$ and the quantum energy $E_{N}^{\rm Q}$ are close. Here, $E^{\rm NLS}$ is given by
\begin{equation}\label{energy:nls-two-component}
E^{\rm NLS} := \inf_{\substack{u_{1},u_{2}\in H^{1}(\mathbb R^{2})\\ \|u_{1}\|_{L^{2}} = 1 = \|u_{2}\|_{L^{2}}}}\mathcal{E}^{\rm NLS}(u_{1},u_{2}).
\end{equation}
Note that $\|\nabla u_{i}\|_{L^{2}} \geq \|\nabla |u_{i}|\|_{L^{2}}$, for any $u_{i} \in H^{1}(\mathbb R^{2})$ and for $i\in\{1,2\}$ (see \cite[Theorem 7.8]{LieLos-01}). Therefore, $\mathcal{E}^{\rm NLS}(u_{1},u_{2}) \geq \mathcal{E}^{\rm NLS}(|u_{1}|,|u_{2}|)$ and we can restrict the minimization problem \eqref{energy:nls-two-component} to non-negative functions. In particular, a ground state for $E^{\rm NLS}$ in \eqref{energy:nls-two-component}, when it exists, can be chosen to be non-negative. 

From now on, we always assume that $0<a_{1},a_{2}<a_{*}$ where $a_{*}>0$ is the optimal constant of the Gagliardo--Nirenberg inequality
\begin{equation}\label{ineq:GN-0} 
\left(\int_{\mathbb R^{2}}|\nabla u(x)|^{2}{\rm d}x\right)
\left(\int_{\mathbb R^{2}}|u(x)|^{2}{\rm d}x\right) \geq\frac{a_{*}}{2}\int_{\mathbb R^{2}}|u(x)|^{4}{\rm d}x,\quad  \forall u\in H^1(\mathbb R^{2}).
\end{equation}
Equivalently, $a_{*} = \|Q\|_{L^{2}}^{2}$ where $Q$ is the unique (up to translations) symmetric radial decreasing positive solution of the equation
\begin{equation}\label{eq:GN}
-\Delta Q+Q-Q^3 = 0 \quad \text{in } \mathbb{R}^{2}.
\end{equation}
It is well-known (see \cite{Weinstein-83,Zhang-00,Maeda-10,GuoSei-14}) that $Q$ is the unique (up to dilations and translations) optimizer for the
inequality \eqref{ineq:GN-0}. One can easily seen from \eqref{ineq:GN-0} and \eqref{eq:GN} that
$$
\|\nabla Q\|_{L^{2}}^{2} = \frac{1}{2}\|Q\|_{L^{4}}^{4} = \|Q\|_{L^{2}}^{2} = a_{*}.
$$
Actually, it was proved in \cite{GuoZenZho-17-dcds,GuoZenZho-18} that \eqref{energy:nls-two-component} admits a minimizer if $0<a_{1},a_{2}<a_{*}$ and either $0<a_{12}<\sqrt{c_{1}^{-1}c_{2}^{-1}(a_{*}-a_{1})(a_{*}-a_{2})}$ or $a_{12}<0$. Furthermore, $E^{\rm NLS} = -\infty$ if either $a_{1}>a_{*}$ or $a_{2}>a_{*}$ or $a_{12}>2^{-1}c_{1}^{-1}c_{2}^{-1}(a_{*}-c_{1}a_{1}-c_{2}a_{2})$. Therefore, $a_{*}$ is the critical interaction strength for the stability of the focusing two-component NLS functional \eqref{functional:nls-two-component}. The blow-up profile of the NLS energy \eqref{energy:nls-two-component} as well as its ground states was established by Guo, Zeng and Zhou in \cite{GuoZenZho-17-dcds,GuoZenZho-18} (see also Section \ref{sec:blow-up-gp} for a review). The purpose of the present paper is to investigate the blow-up behavior of the full many-body system \eqref{hamiltonian}, which is more difficult. We remark that there might be blow up if one of the intra-species interactions is repulsive or if both of them are and the inter-species interactions are attractive. However, the arguments in the present paper do not allow us to cover those cases.

Our work and method are inspired by Lewin, Nam and Rougerie \cite{LewNamRou-17-proc}. In the mentioned paper, the authors studied the collapse of the many-body system arising in a one-component BEC with an attractive interaction (see \cite{GuoSei-14} for the study in the one-body theory and see \cite{Nguyen-19b} for a related topic). In that one-component setting, we remark that $a_{*}$ is also the critical interaction strength for the existence of a ground state for the focusing one-component NLS functional. In addition, the convergence of the many-body ground states was proved for the \emph{single} one-particle reduced density matrices. The two-component BEC presents more complicated phenomena than a single-component BEC since there are inter-species interactions between two components. In our two-component case, the convergence of the many-body ground states will be formulated using the \emph{double} reduced density matrices. Depending on the inter-species interactions, we will discuss the blow-up behavior of the ground state energy \eqref{energy:quantum} and its ground states as well. The precise statements of our results are represented in the next section. The remainder of the paper is then devoted to their proofs.

\section{Main Results}

\subsection{The Case of Attractive Inter-Species Interactions} \label{subsec:blow-up-I}

In the first part of this paper, we consider the totally attractive system, i.e., $a_{12}>0$. In that case, the existence of a ground state for \eqref{energy:nls-two-component}, under the assumptions that $0<a_{1},a_{2}<a_{*}$ and $0<a_{12}<\sqrt{c_{1}^{-1}c_{2}^{-1}(a_{*}-a_{1})(a_{*}-a_{2})}$, follows the standard direct method in the calculus of variations. Note that the last condition on $a_{12}$ enters to control the inter-species interaction. Furthermore, there are no ground states when either $a_{1} \geq a_{*}$ or $a_{2} \geq a_{*}$ or $a_{12} \geq 2^{-1}c_{1}^{-1}c_{2}^{-1}(a_{*}-c_{1}a_{1}-c_{2}a_{2})$. The existence theory for NLS ground states is then left open in the case 
\begin{equation}\label{gap:existence-nls-ground-state}
0<a_{1},a_{2}<a_{*} \text{ and } \sqrt{c_{1}^{-1}c_{2}^{-1}(a_{*}-a_{1})(a_{*}-a_{2})} \leq a_{12} \leq 2^{-1}c_{1}^{-1}c_{2}^{-1}(a_{*}-c_{1}a_{1}-c_{2}a_{2}).
\end{equation}
Under some additional assumptions on $(a_{1},a_{2},a_{12})$, there may exist a ground state. On the other hand, we always have that
\begin{equation}\label{gap:existence-nls-ground-state-minima}
\inf_{x\in\mathbb R^{2}}V_{1}(x) + \inf_{x\in\mathbb R^{2}}V_{2}(x) \leq E^{\rm NLS} \leq \inf_{x\in\mathbb R^{2}}(V_{1}(x)+V_{2}(x)).
\end{equation}
In particular, let $c_{1}(a_{*}-a_{1}) = c_{2}(a_{*}-a_{2})$ and consider \eqref{energy:nls-two-component} at the threshold point $(a_{1},a_{2},a_{12}) = (a^{*}-c_{2}a_{12},a^{*}-c_{1}a_{12},a_{12})$. Then, there exists at least one ground state if $E^{\rm NLS} < \inf_{x\in\mathbb R^{2}}(V_{1}(x)+V_{2}(x))$. On the other hand, if equality holds in the chain of inequalities in \eqref{gap:existence-nls-ground-state-minima}, then there are no ground states (see \cite[Theorems 1.2 and 1.3]{GuoZenZho-17-dcds} for further discussions). Thus, it is reasonable to study the blow-up behavior of the NLS ground states when $V_{1}$ and $V_{2}$ have a common minimum point since this is precisely when equality in \eqref{gap:existence-nls-ground-state-minima} occurs. In that case, we will fix $0<a_{12}<a_{*}\min\{c_{1}^{-1},c_{2}^{-1}\}$ and we take $(a_{1},a_{2}):=(a_{1,N},a_{2,N})\nearrow (a_{*}-c_{2}a_{12},a_{*}-c_{1}a_{12})$ as $N\to\infty$. The two components of the NLS ground states must blow up at the center of the trap and with the same rate. For the detailed analysis, we refer the reader to the paper \cite{GuoZenZho-17-dcds} (see also Subsection \ref{subsec:proof-I} for a review).

In this paper, we study the collapse of the full many-body system \eqref{hamiltonian}. When the inter-species interactions are attractive, we study its ground states in the regime where the total interaction strength of intra-species and inter-species tends to the critical value $a_{*}$ sufficiently slowly. We prove that the many-body ground states are fully condensed on the (unique) Gagliardo--Nirenberg solution \eqref{eq:GN}. In our two-component setting, the convergence of ground states will be formulated using the \emph{double} $(k,\ell)$-particle reduced density matrices. It is defined, for any $\Psi_{N}\in\mathcal{H}_{N}=\mathcal{H}_{N_{1}}\otimes\mathcal{H}_{N_{2}}$, by a partial trace
$$
\gamma_{\Psi_{N}}^{(k,\ell)}:=\tr_{k+1\to N_{1}}\otimes\tr_{\ell+1\to N_{2}}|\Psi_{N}\rangle\langle\Psi_{N}|,\quad \forall k,\ell \in \mathbb{N}.
$$
Equivalently, $\gamma_{\Psi_{N}}^{(k,\ell)}$ is the trace-class operator on $\mathcal{H}_{k}\otimes\mathcal{H}_\ell$ with kernel
$$
\gamma_{\Psi_{N}}^{(k,\ell)}(X,Y;X',Y') = \int_{\mathbb{R}^{2(N_{1}-k)}}\int_{\mathbb{R}^{2(N_{2}-\ell)}} \Psi_{N}(X,Z;Y,T) \overline{\Psi_{N}(X',Z;Y',T)}{\rm d}T{\rm d}Z
$$
where $X,X'\in (\mathbb{R}^{2})^{k}$ and $Y,Y'\in (\mathbb{R}^{2})^{\ell}$. To make the analysis precise, let us assume that the external potentials $V_{1}$ and $V_{2}$ are of the typical forms
\begin{equation}\label{potentials:external}
V_{i}(x) = |x-z_{i}|^{p_{i}}, \quad i\in\{1,2\},
\end{equation}
where $z_{i} \in \mathbb R^{2}$ and $p_{i} > 0$. This is a generalization of the harmonic trapping potentials commonly used in laboratory experiments. Furthermore, let us introduce the following notation
\begin{equation}\label{GN:normalized}
Q_{0} = (a_{*})^{-\frac{1}{2}}Q
\end{equation}
where $Q$ is the (unique) Gagliardo--Nirenberg solution of \eqref{eq:GN}, normalized to have unit $L^{2}(\mathbb R^2)$-norm.

In the case $a_{12}>0$ and $z_{1} \equiv z_{2}$, our first main result is the following.

\begin{theorem}\label{thm:blow-up-bec-1-2-I} 
	Assume that $0<a_{12}<a_{*}\min\{c_{1}^{-1},c_{2}^{-1}\}$ is fixed and $V_{1}$, $V_{2}$ are defined as in \eqref{potentials:external} with $z_{1} = 0 = z_{2}$. Let $0<\beta<1/2$ and let $(a_{1},a_{2}) := (a_{1,N},a_{2,N})\nearrow (a_{*}-c_{2}a_{12},a_{*}-c_{1}a_{12})$ such that $a_{N} := c_{1}a_{1,N} + c_{2}a_{2,N} + 2c_{1}c_{2}a_{12} = a_{*}-N^{-\gamma}$ with
	$$
	0<\gamma<\min\left\{\frac{p_{0}+2}{p_{0}+3}\beta,\frac{p_{0}+2}{p_{0}}(1-2\beta)\right\},\quad p_{0}=\min\{p_{1},p_{2}\}.
	$$
	Let $\Psi_{N}$ be a ground state for $H_{N}$. Let $\Phi_{N} = \ell_{N}^{-N}\Psi_{N}(\ell_{N}^{-1}\cdot)$ where $\ell_{N}=\Lambda (a_{*}-a_{N})^{-\frac{1}{p_{0}+2}}$ with
	\begin{equation}\label{lambda:1-2-I}
	\Lambda = \left(\frac{p_{0}\nu}{2}\int_{\mathbb R^{2}}|x|^{p_{0}}|Q(x)|^{2}{\rm d}x\right)^{\frac{1}{p_{0}+2}} \quad \text{and} \quad \nu = \lim_{x\to 0} \frac{c_{1}V_{1}(x)+c_{2}V_{2}(x)}{|x|^{p_{0}}}.
	\end{equation}
	Then, up to extraction of a subsequence, we have
	\begin{equation}\label{conv:partial-trace-1-2-I}
	\lim_{N\to\infty} \tr\Big|\gamma_{\Phi_{N}}^{(k,\ell)} - \big|Q_{0}^{\otimes k}\otimes Q_{0}^{\otimes \ell}\big\rangle\big\langle Q_{0}^{\otimes k}\otimes Q_{0}^{\otimes \ell}\big|\Big| = 0, \quad \forall k,\ell\in \mathbb{N},
	\end{equation}	
	where $Q_{0}$ is given by \eqref{GN:normalized}. In addition, we have
	\begin{equation}\label{asymptotic:quantum-energy}
	E_{N}^{\rm Q} = E^{\rm NLS} + o(E^{\rm NLS})_{N\to\infty} = (a_{*}-a_{N})^{\frac{p_{0}}{p_{0}+2}} \left(\frac{p_{0}+2}{p_{0}} \cdot \frac{\Lambda^{2}}{a_{*}}+o(1)_{N\to\infty}\right).
	\end{equation}
\end{theorem}

\begin{remark}
	The condition $0<\gamma < \frac{p_{0}+2}{p_{0}}(1-2\beta)$ implies that we consider mean-field interactions. This corresponds to a high-density regime where the particles meet frequently but interact weakly since the typical interaction length is larger than the average distance  between the particles. On the other hand, the condition $\gamma < \frac{p_{0}+2}{p_{0}+3}\beta$ ensures that the Hartree and NLS energies are close in the limit $N \to\infty$.
\end{remark}

There exists another setting for which it is reasonable to study the blow-up behavior of ground states in the case of attractive inter-species interactions: Fix $0<a_{1},a_{2}<a_{*}$ and take $a_{12}:=\alpha_{N} \nearrow \alpha_{*}$ as $N\to\infty$, for some critical value $0 < \alpha_{*} < 2^{-1}c_{1}^{-1}c_{2}^{-1}(a_{*}-c_{1}a_{1}-c_{2}a_{2})$. Since there is a gap in the existence theory for NLS ground states, we will consider only the case $c_{1}(a_{*}-a_{1}) = c_{2}(a_{*}-a_{2})$. Hence, there will be no more discussion on \eqref{gap:existence-nls-ground-state}. We have the following.

\begin{theorem}\label{thm:blow-up-bec-12-I} 
	Assume that $0<a_{1},a_{2}<a_{*}$ are fixed such that $c_{1}(a_{*}-a_{1}) = c_{1}c_{2}\alpha_{*} = c_{2}(a_{*}-a_{2})$ and $V_{1}$, $V_{2}$ are defined as in \eqref{potentials:external} with $z_{1} = 0 = z_{2}$. Let $0<\beta<1/2$ and let $0 < a_{12} := \alpha_{N} = \alpha_{*}-N^{-\gamma}$ with
	$$
	0<\gamma<\min\left\{\frac{p_{0}+2}{p_{0}+3}\beta,\frac{p_{0}+2}{p_{0}}(1-2\beta)\right\}, \quad p_{0}=\min\{p_{1},p_{2}\}.
	$$
	Let $\Psi_{N}$ be a ground state for $H_{N}$. Let $\Phi_{N} = \ell_{N}^{-N}\Psi_{N}(\ell_{N}^{-1}\cdot)$ where $\ell_{N}=\Theta (\alpha_{*}-\alpha_{N})^{-\frac{1}{p_{0}+2}}$ with
	\begin{equation}\label{lambda:12-I}
	\Theta = \left(\frac{p_{0}\nu}{4c_{1}c_{2}}\int_{\mathbb R^{2}}|x|^{p_{0}}|Q(x)|^{2}{\rm d}x\right)^{\frac{1}{p_{0}+2}} \quad \text{and} \quad \nu = \lim_{x\to 0} \frac{c_{1}V_{1}(x)+c_{2}V_{2}(x)}{|x|^{p_{0}}}.
	\end{equation}
	Then, up to extraction of a subsequence, we have
	\begin{equation}\label{conv:partial-trace-12-I}
	\lim_{N\to\infty} \tr\Big|\gamma_{\Phi_{N}}^{(k,\ell)} - \big|Q_{0}^{\otimes k}\otimes Q_{0}^{\otimes \ell}\big\rangle\big\langle Q_{0}^{\otimes k}\otimes Q_{0}^{\otimes \ell}\big|\Big| = 0, \quad \forall k,\ell\in \mathbb{N},
	\end{equation}	
	where $Q_{0}$ is given by \eqref{GN:normalized}. In addition, we have
	$$
	E_{N}^{\rm Q} = E^{\rm NLS} + o(E^{\rm NLS})_{N\to\infty} = (\alpha_{*}-\alpha_{N})^{\frac{p_{0}}{p_{0}+2}} \left(2c_{1}c_{2}\frac{p_{0}+2}{p_{0}} \cdot \frac{\Theta^{2}}{a_{*}}+o(1)_{N\to\infty}\right).
	$$
\end{theorem}

\begin{remark}
	\begin{itemize}
		\item The case $c_{1}(a_{*}-a_{1}) = c_{2}(a_{*}-a_{2})$ is also included in Theorem \ref{thm:blow-up-bec-1-2-I}. Note that, with this assumption, Theorem 1.1 in \cite{GuoZenZho-17-dcds} gave a complete classification of the existence and non-existence of ground states for \eqref{energy:nls-two-component}.
		\item To avoid the assumption $c_{1}(a_{*}-a_{1}) = c_{2}(a_{*}-a_{2})$, one might consider a more evolved NLS model, where the constraint condition in \eqref{energy:nls-two-component} is replaced by $\|u_{1}\|_{L^{2}}^{2} + \|u_{2}\|_{L^{2}}^{2} = 1$ (see, e.g., \cite{BaoCai-11,GuoLiWeiZen-19a,GuoLiWeiZen-19b}). However, the many-body theory behind this is still an open problem. We hope to come back to this issue in the future.
	\end{itemize}
\end{remark}

\subsection{The Case of Repulsive Inter-Species Interactions} \label{subsec:blow-up-II}
In the second part, we consider the system with attractive intra-species interactions and repulsive inter-species interactions, i.e., $a_{12}<0$. In that case, the existence of ground states for \eqref{energy:nls-two-component}, under the assumptions that $0<a_{1},a_{2}<a_{*}$ and $a_{12}<0$ is fixed, follows the standard direct method in the calculus of variations. Furthermore, when either $a_{1} \geq a_{*}$ or $a_{2} \geq a_{*}$ or $a_{1} = a_{*} = a_{2}$, there are no ground states for \eqref{energy:nls-two-component}. The limit behavior of the NLS energy as well as its ground states, when $a_{12}<0$ is fixed and $(a_{1},a_{2}):=(a_{1,N},a_{2,N}) \nearrow (a_{*},a_{*})$ as $N\to\infty$, has been analyzed in \cite{GuoZenZho-18} (see Subsection \ref{subsec:proof-II} for a review). This is somehow similar to the one-component setting where the minimization problems read
\begin{equation}\label{energy:nls-one-component}
E_{i}^{\rm NLS} = \inf_{\substack{u_{i}\in H^{1}(\mathbb R^2)\\ \|u_{i}\|_{L^{2}}=1}}\mathcal{E}_{i}^{\rm NLS}(u_{i}),\quad i\in\{1,2\}.
\end{equation}
Here, the NLS energy functionals are given by
\begin{equation}\label{functional:nls-one-component}
\mathcal{E}_{i}^{\rm NLS}(u_{i}) = \int_{\mathbb R^{2}}\left[|\nabla u_{i}(x)|^{2}+V_{i}(x)|u_{i}(x)|^{2}-\frac{a_{i}}{2}|u_{i}(x)|^{4}\right]{\rm d}x,\quad i\in\{1,2\}.
\end{equation}
For the reader's convenience, let us briefly recall the known results concerning the blow-up behavior of $E_{i}^{\rm NLS}$ in \eqref{energy:nls-one-component} as well as its ground states. From \cite{GuoSei-14} we have, for $i\in\{1,2\}$,
\begin{equation}\label{lambda:II}
\lim_{a_{i}\nearrow a_{*}}\frac{E_{i}^{\rm NLS}}{(a_{*}-a_{i})^{\frac{p_{i}}{p_{i}+2}}} = \frac{p_{i}+2}{p_{i}} \cdot \frac{\Lambda_{i}^{2}}{a_{*}} \quad \text{where} \quad \Lambda_{i} = \left(\frac{p_{i}}{2}\int_{\mathbb R^{2}}|x|^{p_{i}}|Q(x)|^{2}{\rm d}x\right)^{\frac{1}{p_{i}+2}}.
\end{equation}
Here, we made use of potentials $V_{i}$, for $i\in\{1,2\}$, which are given by \eqref{potentials:external}.  In addition, assume that $u_{i}$ is a positive ground state for $E_{i}^{\rm NLS}$ in \eqref{energy:nls-one-component} for each $0<a_{i}<a_{*}$. Then, up to extraction of a subsequence, we have
$$
\lim_{a_{i} \nearrow a_{*}} \Lambda_{i}^{-1}(a_{*}-a_{i})^{\frac{1}{p_{i}+2}} u_{i}(\Lambda_{i}^{-1}(a_{*}-a_{i})^{\frac{1}{p_{i}+2}}x+z_{i}) = Q_{0}(x).
$$
strongly in $H^{1}(\mathbb R^2)$, where $z_{i}$ is the minima of $V_{i}$ and $Q_{0}$ is given by \eqref{GN:normalized}.

In this paper, we study the collapse of the full many-body system \eqref{hamiltonian}. When the inter-species interactions are repulsive, we study its ground states in the regime where the interaction strength of intra-species among particles in each component tends to the critical value $a_{*}$ sufficiently slowly. We prove that the many-body ground states are fully condensed on the (unique) Gagliardo--Nirenberg solution \eqref{eq:GN}. Our last main result is the following.

\begin{theorem}\label{thm:blow-up-bec-II} 
	Assume that $a_{12}<0$ is fixed and $V_{1}$, $V_{2}$ are defined as in \eqref{potentials:external} with $z_{1}\ne z_{2}$. Let $0<\beta<1/2$ and let $a_{i} := a_{i,N}=a_{*}-N^{-\gamma_{i}}$, for $i\in\{1,2\}$, with
	$$
	\frac{\gamma_{1}}{\gamma_{2}} = \frac{p_{1}+2}{p_{2}+2} \quad \text{and} \quad 0<\gamma_{i}<\min\left\{\frac{p_{i}+2}{p_{i}+3}\beta,\frac{p_{i}+2}{p_{i}}(1-2\beta)\right\}.
	$$
	Let $\Psi_{N}$ be a ground state for $H_{N}$. Let
	$$
	\Phi_{N}(x_{1},\ldots,x_{N_{1}};y_{1},\ldots,y_{N_{2}}) = \frac{\Psi_{N}\left(\cfrac{x_{1}}{\ell_{1,N}}+z_{1},\ldots,\cfrac{x_{N_{1}}}{\ell_{1,N}}+z_{1};\cfrac{y_{1}}{\ell_{2,N}}+z_{2},\ldots,\cfrac{y_{N_{2}}}{\ell_{2,N}}+z_{2}\right)}{\ell_{1,N}^{N_{1}}\ell_{2,N}^{N_{2}}}
	$$
	where $\ell_{i,N}=\Lambda_{i} (a_{*}-a_{i,N})^{-\frac{1}{p_{i}+2}}$, for $i\in\{1,2\}$, with $\Lambda_{i}$ are given by \eqref{lambda:II}.
	Then, up to extraction of a subsequence, we have
	\begin{equation}\label{conv:partial-trace-II}
	\lim_{N\to\infty} \tr\Big|\gamma_{\Phi_{N}}^{(k,\ell)} - \big|Q_{0}^{\otimes k}\otimes Q_{0}^{\otimes \ell}\big\rangle\big\langle Q_{0}^{\otimes k}\otimes Q_{0}^{\otimes \ell}\big|\Big| = 0, \quad \forall k,\ell\in \mathbb{N},
	\end{equation}	
	where $Q_{0}$ is given by \eqref{GN:normalized}. In addition, with $E_{1}^{\rm NLS}$ and $E_{2}^{\rm NLS}$ defined in \eqref{energy:nls-one-component}, we have
	$$
		E_{N}^{\rm Q} = \sum_{i=1}^{2}c_{i}E_{i}^{\rm NLS} + o(E_{i}^{\rm NLS}) = \sum_{i=1}^{2}c_{i}(a_{*}-a_{i,N})^{\frac{p_{i}}{p_{i}+2}}\left(\frac{p_{i}+2}{p_{i}} \cdot \frac{\Lambda_{i}^{2}}{a_{*}}+o(1)_{N\to\infty}\right).
		$$
\end{theorem}

\begin{remark}
	\begin{itemize}
		\item The condition $\frac{\gamma_{1}}{\gamma_{2}} = \frac{p_{1}+2}{p_{2}+2}$ is a technical assumption which yields that $\ell_{1,N}$ and $\ell_{2,N}$ have the same asymptotic behavior when $N\to\infty$. This will be used only to prove the convergence of ground states in \eqref{conv:partial-trace-II}, but not the asymptotic behavior of the quantum energy.
		\item The convergence of density matrices follows from that of the quantum energy. In the case $z_{1} = z_{2}$, however, we are not able to give an asymptotic formula for the quantum energy.
	\end{itemize}
\end{remark}

Note that, for $i\in\{1,2\}$, the positive ground states $u_{a_{i}}$ of \eqref{energy:nls-one-component} decay exponentially (see \cite[Proposition A]{GuoZenZho-18}). More precisely, for any $R > 0$, there exists $C(R) > 0$ such that
\begin{equation}\label{decay:exponential}
u_{a_{i}}(x) \leq C(R)e^{-\mu|x-z_{i}|(a_{*}-a_{i})^{-\frac{1}{p_{i}+2}}} \quad \text{in } \mathbb R^{2} \setminus B_{R}(x_{i})
\end{equation}
where $\mu > 0$ is independent of $R$, $z_{i}$ and $a_{i}$. The decay property \eqref{decay:exponential} is used to study the blow-up profile of ground states for \eqref{energy:nls-two-component}. In fact, as pointed out in \cite{GuoZenZho-18}, we are not able to give the optimal energy estimate for the NLS energy because of the presence of the cross-term $-a_{12}\int_{\mathbb R^{2}}|u_{1}(x)|^{2}|u_{2}(x)|^{2}{\rm d}x$ in \eqref{functional:nls-two-component}. In the energy estimate, if $z_{1}\ne z_{2}$, then the cross-term can be made arbitrary small in the limit regime $(a_{1},a_{2}) \nearrow (a_{*},a_{*})$, by \eqref{decay:exponential}. Thus, we can give the refined estimate for the limit behavior of ground states. However, such a refined calculation for the NLS energy is not known when $z_{1} \equiv z_{2}$. In that case, it was proved that the NLS ground states blow up, but one cannot determine the accurate blow-up rate (see \cite[Theorem 1.3]{GuoZenZho-18}).

\subsection{Methodology of the Proofs} 
Similarly to what was done in \cite{LewNamRou-17-proc} for the one-component setting, our proofs of BEC in this paper are based on a Feynman--Hellmann-type argument. It relies strongly on the uniqueness of the limiting profile for NLS ground states, i.e., the unique positive solution of \eqref{eq:GN}. To obtain the convergences of ground states in Theorems \ref{thm:blow-up-bec-1-2-I}, \ref{thm:blow-up-bec-12-I} and \ref{thm:blow-up-bec-II}, it is enough to prove \eqref{conv:partial-trace-1-2-I},  \eqref{conv:partial-trace-12-I}  and \eqref{conv:partial-trace-II} for the double indices $(1,0)$ and $(0,1)$. We refer the reader to \cite[Section 3]{MicOlg-17} for the general discussion. The main difficulty in this paper, as well as in \cite{LewNamRou-17-proc}, is the energy estimate between the quantum and the NLS energies via the Hartree energy. In the next section, we will show that, under the intra-species interactions and inter-species interaction given by \eqref{assumption:potential-1}, we have
\begin{equation}\label{conv:quantum-hartree-energy}
E_{N}^{\rm H} \geq E_{N}^{\rm Q} \geq E_{N}^{\rm H} - CN^{2\beta-1}.
\end{equation}
While the upper bound is trivial, by the variational principle, it is more complicated to obtain the lower bound in \eqref{conv:quantum-hartree-energy}. Our strategy is to adapt the arguments in \cite[Section 3]{Lewin-ICMP} to our two-component setting. Next, we compare the Hartree and NLS energies and show that the error term is $N^{-\beta}$, thanks to assumption \eqref{assumption:potential-2}. We arrive at the final estimate
\begin{equation}\label{conv:quantum-nls-energy}
E^{\rm NLS} + CN^{-\beta} \geq E_{N}^{\rm Q} \geq E^{\rm NLS} - CN^{-\beta} - CN^{2\beta-1}.
\end{equation}
The NLS  energy functional is stable, i.e., $E^{\rm NLS} \geq 0$, under the assumptions that $0<a_{1},a_{2}<a_{*}$ and either $0<a_{12}<\sqrt{c_{1}^{-1}c_{2}^{-1}(a_{*}-a_{1})(a_{*}-a_{2})}$ or $a_{12}<0$ (see \cite{GuoZenZho-17-dcds,GuoZenZho-18}). In addition,  if $0<\beta<1/2$, then \eqref{conv:quantum-nls-energy} implies the convergence of the quantum energy to the NLS  energy. The novelty of our work in the present paper is to prove that the above condition between $a_{1}$, $a_{2}$ and $a_{12}$, which yields the stability of the NLS  energy functional \eqref{functional:nls-two-component}, is sufficient for the stability of the many-body system \eqref{hamiltonian} as well.

It it expected that the convergence of the quantum energy to the NLS  energy holds true for any $0<\beta<1$. However, a proof for $1/2\leq \beta <1$ is much more involved (see \cite{LewNamRou-14,LewNamRou-16,LewNamRou-17,NamRou-20} for discussions in the one-component case). We hope that our study in this paper can serve as a first step for understanding the 2D focusing mixture Bose gases.

\medskip

\noindent\textbf{Organization of the paper.} We now describe the structure of this paper. In Section \ref{sec:quantum-nls-energy}, we prove the convergence of the quantum energy to the Hartree energy for a more general system. In Section \ref{sec:blow-up-gp}, we give proofs of Theorems \ref{thm:blow-up-bec-1-2-I}, \ref{thm:blow-up-bec-12-I} and \ref{thm:blow-up-bec-II} after revisiting the blow-up phenomenon in the NLS theory and establishing energy estimates for the quantum energy.

\section{From the Quantum Energy to the NLS Energy}\label{sec:quantum-nls-energy}

\subsection{Convergence of the Quantum Energy to the Hartree Energy}\label{subsec:quantum-hartree-energy}

In this subsection, we prove the convergence of the quantum energy to the Hartree energy under some assumptions on the kinetic and the potentials energies. We consider the general Hamiltonian
\begin{align*}
H_{N} = &  \sum_{i=1}^{N_{1}}h^{(1)}_{x_{i}} + \frac{1}{N_{1}-1}\sum_{1 \leq i < j \leq N_{1}}W^{(1)}(x_{i}-x_{j})  \\
& + \sum_{r=1}^{N_{2}}h^{(2)}_{y_{r}} + \frac{1}{N_{2}-1}\sum_{1 \leq r < s \leq N_{2}}W^{(2)}(y_{r}-y_{s}) \\
& + \frac{1}{N}\sum_{i=1}^{N_{1}}\sum_{r=1}^{N_{2}}W^{(12)}(x_{i}-y_{r})
\end{align*}
which is the many-body system for $N_{1}$ and $N_{2}$ identical bosons of different types in $\mathbb R^2$, acting on the Hilbert space $\mathcal{H}_{N}$ given by \eqref{space:hilbert}. Here, we denote by $N = N_{1} + N_{2}$ the total number of particles. The kinetic energies $h^{(1)}$ and $h^{(2)}$, which are the non-interacting one-body Hamiltonian, are assumed to be real and positive preserving, i.e.,  $\langle u,h^{(i)}u\rangle \geq \langle |u|,h^{(i)}|u|\rangle$ for $i\in\{1,2\}$. The two-body interactions among particles of the same species $W^{(1)}$, $W^{(2)}$ and of different species $W^{(12)}$ satisfy $\widehat{W^{(\sigma)}}\in L^1(\mathbb R^{2})$, for $\sigma\in\{1,2,12\}$. For the sake of simplicity, we assume that the ratios $c_{1} = N_{1}/N$ and $c_{2} = N_{1}/N$ are fixed. The ground state energy per particle and the Hartree energy are given by
$$
E_{N}^{\rm Q} = N^{-1}\inf\sigma_{\mathcal{H}_{N}}H_{N}  \quad \text{and} \quad  E^{\rm H} = \inf_{\substack{u_{1},u_{2}\in H^{1}(\mathbb R^{2})\\ \|u_{1}\|_{L^{2}} = 1 = \|u_{2}\|_{L^{2}}}}\mathcal{E}^{\rm H}(u_{1},u_{2}).
$$
Here, the Hartree energy functional, which is obtained by considering the ansatz $u_{1}^{\otimes N_{1}}\otimes u_{2}^{\otimes N_{2}}$ as a trial wave function for the many-body system, is given by
\begin{align*}
\mathcal{E}^{\rm H} (u_{1},u_{2}) & = c_{1}\left[\langle u_{1},h^{(1)}u_{1} \rangle + \frac{1}{2}\int_{\mathbb R^{2}}|u_{1}(x)|^{2}(W^{(1)}\star|u_{1}|^{2})(x){\rm d}x\right] \\
& \quad  + c_{2}\left[\langle u_{2},h^{(2)}u_{2} \rangle + \frac{1}{2}\int_{\mathbb R^{2}}|u_{2}(x)|^{2}(W^{(2)}\star|u_{2}|^{2})(x){\rm d}x\right]\\
& \quad  + c_{1}c_{2}\int_{\mathbb R^{2}}|u_{1}(x)|^{2}(W^{(12)}\star |u_{2}|^{2})(x){\rm d}x.
\end{align*}

\begin{theorem}\label{thm:quantum-hartree-energy-general}
	Under previous assumptions, we have
	\begin{equation}\label{conv:quantum-hartree-energy-general}
	E^{\rm H} \geq E_{N}^{\rm Q} \geq E^{\rm H} - \frac{C}{N}\left( \|\widehat{W^{(1)}}\|_{L^1} + \|\widehat{W^{(2)}}\|_{L^1} + \|\widehat{W^{(12)}}\|_{L^1}\right).
	\end{equation}
\end{theorem}

We remark that, based on quantum de Finetti theorem, the convergence \eqref{conv:quantum-hartree-energy-general} has been proven in \cite[Theorem 4.1]{MicNamOlg-19} for confined systems without convergence rate. Note that, in Theorem \ref{thm:quantum-hartree-energy-general}, we made only the assumption on the positivity preserving of the kinetic energy. Furthermore, we did not make any assumption on the sign of $W^{(\sigma)}$, for $\sigma\in\{1,2,12\}$, nor on its Fourier transform $\widehat{W^{(\sigma)}}$. The intra-species and inter-species interactions can be either attractive or repulsive. If $\widehat{W^{(\sigma)}}$ are not integrable (e.g., for Coulomb potentials), the proof can be done by an approximation argument. To prove Theorem \ref{thm:quantum-hartree-energy-general}, we will need the following variant of Onsager's lemma \cite{Onsager-39} which relies on the estimate of the two-body interaction by a one-body term.

\begin{lemma}\label{lem:two-one}
	Given the functions $V^{(1)}$, $V^{(2)}$ and $V^{(12)}$ in $\mathbb R^2$. Assume that the Fourier transforms $\widehat{V^{(1)}}$, $\widehat{V^{(2)}}$, $\widehat{V^{(12)}}$ are positives and belong to $L^1(\mathbb R^{2})$. Then, for any integrable functions $\chi$ and $\zeta$, we have
	\begin{align}
	\sum_{1 \leq i < j \leq N_{1}}V^{(1)}(x_{i}-x_{j}) & \geq - \frac{1}{2}\iint_{\mathbb R^{2} \times \mathbb R^{2}}\chi(x)\chi(y)V^{(1)}(x-y){\rm d}x{\rm d}y \nonumber \\
	& \quad + \sum_{i=1}^{N_{1}}\big(\chi\star V^{(1)}\big)(x_{i}) - \frac{N_{1}}{2}V^{(1)}(0), \label{two-one-1}\\
	\sum_{1 \leq r < s \leq N_{2}}V^{(2)}(y_{r}-y_{s}) & \geq - \frac{1}{2}\iint_{\mathbb R^{2} \times \mathbb R^{2}}\zeta(x)\zeta(y)V^{(2)}(x-y){\rm d}x{\rm d}y \nonumber \\
	& \quad + \sum_{r=1}^{N_{2}}\big(\zeta\star V^{(2)}\big)(y_{r}) - \frac{N_{2}}{2}V^{(2)}(0), \label{two-one-2} \\
	\sum_{i=1}^{N_{1}}\sum_{r=1}^{N_{2}}V^{(12)}(x_{i}-y_{r}) & \geq -\sum_{1 \leq i < j \leq N_{1}}V^{(12)}(x_{i}-x_{j}) + \sum_{i=1}^{N_{1}}\big[(\chi+\zeta)\star V^{(12)}\big](x_{i}) \nonumber \\
	& \quad  - \sum_{1 \leq r < s \leq N_{2}}V^{(12)}(y_{r}-y_{s}) + \sum_{r=1}^{N_{2}}\big[(\chi+\zeta)\star V^{(12)}\big](y_{r}) \nonumber \\
	& \quad  - \frac{1}{2}\iint_{\mathbb R^{2} \times \mathbb R^{2}}\left(\chi(x)\chi(y)+\zeta(x)\zeta(y)\right)V^{(12)}(x-y){\rm d}x{\rm d}y \nonumber \\
	& \quad  - \iint_{\mathbb R^{2} \times \mathbb R^{2}}\chi(x)\zeta(y)V^{(12)}(x-y){\rm d}x{\rm d}y -\frac{N}{2}V^{(12)}(0). \label{two-one-3}
	\end{align}
\end{lemma}

\begin{proof} 
	\eqref{two-one-1}, \eqref{two-one-2} and \eqref{two-one-3} are obtained by expanding
	$$
	\iint_{\mathbb R^{2} \times \mathbb R^{2}} f^{(\sigma)}(x)f^{(\sigma)}(y)V^{(\sigma)}(x-y){\rm d}x{\rm d}y = 2\pi\int_{\mathbb R^{2}}\widehat{V^{(\sigma)}}(k)|\widehat{f^{(\sigma)}}(k)|^{2}{\rm d}k \geq 0,
	$$
	for $\sigma\in\{1,2,12\}$, where
	$$
	f^{(1)} = \sum_{i=1}^{N_{1}}\delta_{x_{i}} - \chi, \quad  f^{(2)} = \sum_{r=1}^{N_{2}}\delta_{y_{r}} - \zeta \quad \text{and} \quad f^{(12)} = f^{(1)} + f^{(2)}.
	$$
\end{proof}

Now, we follow the method described in \cite[Section 3]{Lewin-ICMP} to prove Theorem \ref{thm:quantum-hartree-energy-general} for arbitrary potential $W^{(\sigma)}$ satisfying $\widehat{W^{(\sigma)}}\in L^1(\mathbb R^{2})$, for $\sigma\in\{1,2,12\}$. The idea was in turn inspired by arguments of L\'evy--Leblond \cite{LevLeb-69}. See \cite{LieThi-84,LieYau-87} for related arguments.

\begin{proof}[Proof of Theorem \ref{thm:quantum-hartree-energy-general}]
	We first consider the case with an even number $2N_{1}$ and $2N_{2}$ of particles of different types. We denote by $2N = 2N_{1} + 2N_{2}$ the total number of particles that we split in two groups of $N=N_{1}+N_{2}$. The position of the $N$ first will be denoted by $x_1,\ldots,x_{N_{1}}$ and $y_1,\ldots,y_{N_{2}}$, whereas those of the others will be denoted by $p_1=x_{N_{1}+1},\ldots,p_{N_{1}}=x_{2N_{1}}$ and $q_1=y_{N_{2}+1},\ldots,q_{N_{2}}=y_{2N_{2}}$. Next, we pick a $2N$-particles state $\Psi_{2N}$ and use its bosonic symmetry in two groups of $2N_{1}$ and $2N_{2}$ variables to write
	$$
	\frac{1}{2N}\Big\langle\Psi_{2N}\Big|\sum_{i=1}^{2N_{1}}h^{(1)}_{x_{i}} + \sum_{r=1}^{2N_{2}}h^{(2)}_{y_{r}}\Big|\Psi_{2N}\Big\rangle = \frac{1}{N}\Big\langle\Psi_{2N}\Big|\sum_{i=1}^{N_{1}}h^{(1)}_{x_{i}} + \sum_{r=1}^{N_{2}}h^{(2)}_{y_{r}}\Big|\Psi_{2N}\Big\rangle.
	$$
	Now, we define 
	$$
	W_{N}^{(12)}=\frac{1}{N}W^{(12)},\quad W_{N}^{(1)} = \frac{1}{N_{1}-1}W^{(1)}-W_{N}^{(12)} \quad \text{and} \quad W_{N}^{(2)} = \frac{1}{N_{2}-1}W^{(2)}-W_{N}^{(12)}.
	$$
	For $\sigma\in\{1,2,12\}$, we decompose 
	$$
	W_{N}^{(\sigma)}=W_{N,+}^{(\sigma)}-W_{N,-}^{(\sigma)},
	$$
	where 
	$$
	\widehat{W_{N,+}^{(\sigma)}}=\big(\widehat{W_{N}^{(\sigma)}}\big)_+\geq 0 \quad \text{and} \quad \widehat{W_{N,-}^{(\sigma)}}=\big(\widehat{W_{N}^{(\sigma)}}\big)_-\geq 0.$$
	We write the repulsive part using only the $x_{i}$'s and $y_{r}$'s as follows
	\begin{align*}
	& \begin{multlined}[t] \frac{1}{2N}\Big\langle\Psi_{2N}\Big|\frac{1}{2N_{1}-1}\sum_{1 \leq i < j \leq 2N_{1}}W_{+}^{(1)}(x_{i}-x_{j}) + \frac{1}{2N_{2}-1}\sum_{1 \leq r < s \leq 2N_{2}}W_{+}^{(2)}(y_{r}-y_{s}) \\
	+ \frac{1}{2N}\sum_{i=1}^{2N_{1}}\sum_{r=1}^{2N_{2}}W_{+}^{(12)}(x_{i}-y_{r})\Big|\Psi_{2N}\Big\rangle 
	\end{multlined} \\
	& \quad  =
	\begin{multlined}[t]
	\frac{1}{N}\Big\langle\Psi_{2N}\Big|\frac{1}{N_{1}-1}\sum_{1 \leq i < j \leq N_{1}}W_{+}^{(1)}(x_{i}-x_{j}) + \frac{1}{N_{2}-1}\sum_{1 \leq r < s \leq N_{2}}W_{+}^{(2)}(y_{r}-y_{s}) \\
	+ \frac{1}{N}\sum_{i=1}^{N_{1}}\sum_{r=1}^{N_{2}}W_{+}^{(12)}(x_{i}-y_{r})\Big|\Psi_{2N}\Big\rangle 
	\end{multlined} \\
	& \quad  =
	\begin{multlined}[t]
	\frac{1}{N}\Big\langle\Psi_{2N}\Big|\sum_{1 \leq i < j \leq N_{1}}W_{N,+}^{(1)}(x_{i}-x_{j}) + \sum_{1 \leq r < s \leq N_{2}}W_{N,+}^{(2)}(y_{r}-y_{s}) \\
	+ \sum_{1 \leq i < j \leq N_{1}}W_{N,+}^{(12)}(x_{i}-x_{j}) + \sum_{1 \leq r < s \leq N_{2}}W_{N,+}^{(12)}(y_{r}-y_{s}) \\
	+ \sum_{i=1}^{N_{1}}\sum_{r=1}^{N_{2}}W_{N,+}^{(12)}(x_{i}-y_{r}) \Big|\Psi_{2N} \Big\rangle.
	\end{multlined}
	\end{align*}
	On the other hand, we express the attractive part as the difference of two terms, involving respectively only the $p_{k}$'s and $q_{m}$'s and both groups	
	\begin{align*}	
	& \begin{multlined}[t]
	-\frac{1}{2N}\Big\langle\Psi_{2N}\Big|\frac{1}{2N_{1}-1}\sum_{1 \leq i < j \leq 2N_{1}}W_{-}^{(1)}(x_{i}-x_{j}) + \frac{1}{2N_{2}-1}\sum_{1 \leq r < s \leq 2N_{2}}W_{-}^{(2)}(y_{r}-y_{s}) \\
	+ \frac{1}{2N}\sum_{i=1}^{2N_{1}}\sum_{r=1}^{2N_{2}}W_{-}^{(12)}(x_{i}-y_{r})\Big|\Psi_{2N}\Big\rangle 
	\end{multlined} \\
	& \quad  = \begin{multlined}[t]
	\frac{1}{N}\Big\langle\Psi_{2N}\Big|\sum_{1 \leq k < \ell \leq N_{1}}W_{N,-}^{(1)}(p_{k}-p_\ell) + \sum_{1 \leq m < n \leq N_{2}}W_{N,-}^{(2)}(q_{m}-q_n) \\
	+ \sum_{1 \leq k < \ell \leq N_{1}}W_{N,-}^{(12)}(p_{k}-p_\ell) + \sum_{1 \leq m < n \leq N_{2}}W_{N,-}^{(12)}(q_{m}-q_n) \\
	+ \sum_{k=1}^{N_{1}}\sum_{m=1}^{N_{2}}W_{N,-}^{(12)}(p_{k}-q_{m}) \Big|\Psi_{2N} \Big\rangle 
	\end{multlined} \\
	& \quad  \quad - \begin{multlined}[t] \frac{1}{N}\Big\langle\Psi_{2N}\Big|\frac{1}{N_{1}}\sum_{i=1}^{N_{1}}\sum_{k=1}^{N_{1}}W_{-}^{(1)}(x_{i}-p_{k}) + \frac{1}{N_{2}}\sum_{r=1}^{N_{2}}\sum_{m=1}^{N_{2}}W_{-}^{(2)}(y_{r}-q_{m}) \\
	+ \sum_{i=1}^{N_{1}}\sum_{m=1}^{N_{2}}W_{N,-}^{(12)}(x_{i}-q_{m}) + \sum_{k=1}^{N_{1}}\sum_{r=1}^{N_{2}}W_{N,-}^{(12)}(p_{k}-y_{r})\Big|\Psi_{2N}\Big\rangle.
	\end{multlined}
	\end{align*}
	This means that $\Big\langle \Psi_{2N}\Big|\dfrac{H_{2N}}{2N}\Big|\Psi_{2N}\Big\rangle = \Big\langle \Psi_{2N}\Big|\dfrac{\tilde{H}_{N}}{N}\Big|\Psi_{2N}\Big\rangle$ where 
	\begin{align*}	
	\tilde{H}_{N} & = \sum_{i=1}^{N_{1}}h^{(1)}_{x_{i}} + \sum_{r=1}^{N_{2}}h^{(2)}_{y_{r}} \\
	& \quad  + \sum_{1 \leq i < j \leq N_{1}}W_{N,+}^{(1)}(x_{i}-x_{j}) + \sum_{1 \leq k < \ell \leq N_{1}}W_{N,-}^{(1)}(p_{k}-p_\ell)  \\
	& \quad  + \sum_{1 \leq r < s \leq N_{2}}W_{N,+}^{(2)}(y_{r}-y_{s}) + \sum_{1 \leq m < n \leq N_{2}}W_{N,-}^{(2)}(q_{m}-q_n)  \\
	& \quad  - \frac{1}{N_{1}}\sum_{i=1}^{N_{1}}\sum_{k=1}^{N_{1}}W_{-}^{(1)}(x_{i}-p_{k}) - \frac{1}{N_{2}}\sum_{r=1}^{N_{2}}\sum_{m=1}^{N_{2}}W_{-}^{(2)}(y_{r}-q_{m}) \\
	& \quad   + \sum_{1 \leq i < j \leq N_{1}}W_{N,+}^{(12)}(x_{i}-x_{j}) + \sum_{1 \leq r < s \leq N_{2}}W_{N,+}^{(12)}(y_{r}-y_{s}) \\
	& \quad + \sum_{i=1}^{N_{1}}\sum_{r=1}^{N_{2}}W_{N,+}^{(12)}(x_{i}-y_{r}) + \sum_{k=1}^{N_{1}}\sum_{m=1}^{N_{2}}W_{N,-}^{(12)}(p_{k}-q_{m}) \\
	& \quad  + \sum_{1 \leq k < \ell \leq N_{1}}W_{N,-}^{(12)}(p_{k}-p_\ell) + \sum_{1 \leq m < n \leq N_{2}}W_{N,-}^{(12)}(q_{m}-q_n) \\
	& \quad  - \sum_{i=1}^{N_{1}}\sum_{m=1}^{N_{2}}W_{N,-}^{(12)}(x_{i}-q_{m}) - \sum_{k=1}^{N_{1}}\sum_{r=1}^{N_{2}}W_{N,-}^{(12)}(p_{k}-y_{r}).
	\end{align*}
	The Hamiltonian $\tilde{H}_{N}$ describes a system of $N = N_{1} + N_{2}$ quantum particles that repel through potentials $W_{N,+}^{(\sigma)}$, for $\sigma\in\{1,2,12\}$, and $N = N_{1} + N_{2}$ classical particles that repel through potentials $W_{N,-}^{(\sigma)}$, with repulsion potentials $W_{N,+}^{(12)}$ and attraction potentials $\frac{1}{N_{1}}W_{-}^{(1)}$, $\frac{1}{N_{2}}W_{-}^{(2)}$, $W_{N,-}^{(12)}$ between two groups. In order to bound $\tilde{H}_{N}$ from below, we first fix the positions $p_1,\ldots,p_{N_{1}}$; $q_1,\ldots,q_{N_{2}}$ of the particles in the second group and consider $\tilde{H}_{N}$ as an operator acting only over the $x_{i}$'s and $y_{r}$'s. Let $\Phi_{N}$ be any bosonic $N$-particles state in the $N=N_{1}+N_{2}$ first variables. Let $\gamma_{\Phi_{N}}^{(1,0)}$ and $\gamma_{\Phi_{N}}^{(0,1)}$ be the $(1,0)$ and $(0,1)$-particle reduced density matrices associated with $\Phi_{N}$. Denote by $\rho_{\Phi_{N}}^{(1,0)}(x) = \gamma_{\Phi_{N}}^{(1,0)}(x,x)$ and $\rho_{\Phi_{N}}^{(0,1)}(x) = \gamma_{\Phi_{N}}^{(0,1)}(x,x)$ the density functions associated with $\gamma_{\Phi_{N}}^{(1,0)}$ and $\gamma_{\Phi_{N}}^{(0,1)}$. Applying Lemma \ref{lem:two-one} for the repulsive potential 
		$$
		V^{(1)} = W_{N,+}^{(1)},\quad V^{(2)} = W_{N,+}^{(2)} \quad \text{and} \quad V^{(12)} = W_{N,+}^{(12)}
		$$
		with
		$$
		\chi = N_{1}\rho_{\Phi_{N}}^{(1,0)} \quad \text{and} \quad \zeta = N_{2}\rho_{\Phi_{N}}^{(0,1)}
		$$
		and using the Hoffmann-Ostenhof inequality \cite{Hof-77} (see also \cite[Lemma 3.2]{Lewin-ICMP}), we obtain
	\begin{align}
	\langle \Phi_{N}|\tilde{H}_{N}|\Phi_{N}\rangle &\geq N_{1}\Big\langle \sqrt{\rho_{\Phi_{N}}^{(1,0)}},h^{(1)}\sqrt{\rho_{\Phi_{N}}^{(1,0)}}\Big\rangle + N_{2}\Big\langle \sqrt{\rho_{\Phi_{N}}^{(0,1)}},h^{(2)}\sqrt{\rho_{\Phi_{N}}^{(0,1)}}\Big\rangle \nonumber\\
	& \quad + \frac{N_{1}^{2}}{2}\iint_{\mathbb R^{2} \times \mathbb R^{2}}\rho_{\Phi_{N}}^{(1,0)}(x)\rho_{\Phi_{N}}^{(1,0)}(y)\big(W_{N,+}^{(1)}+W_{N,+}^{(12)}\big)(x-y){\rm d}x{\rm d}y \label{ineq:positive-1}\\
	& \quad  + \frac{N_{2}^{2}}{2}\iint_{\mathbb R^{2} \times \mathbb R^{2}}\rho_{\Phi_{N}}^{(0,1)}(x)\rho_{\Phi_{N}}^{(0,1)}(y)\big(W_{N,+}^{(2)}+W_{N,+}^{(12)}\big)(x-y){\rm d}x{\rm d}y \label{ineq:positive-2}\\
	& \quad + \sum_{1 \leq k < \ell \leq N_{1}}W_{N,-}^{(1)}(p_{k}-p_\ell) + \sum_{1 \leq m < n \leq N_{2}}W_{N,-}^{(2)}(q_{m}-q_n) \label{ineq:negative-1}\\
	& \quad  - \sum_{k=1}^{N_{1}}(\rho_{\Phi_{N}}^{(1,0)}\star W_{-}^{(1)})(p_{k}) - \sum_{m=1}^{N_{2}}(\rho_{\Phi_{N}}^{(0,1)}\star W_{-}^{(2)})(q_{m})  \label{ineq:negative-2}\\
	& \quad + \sum_{1 \leq k < \ell \leq N_{1}}W_{N,-}^{(12)}(p_{k}-p_\ell) + \sum_{1 \leq m < n \leq N_{2}}W_{N,-}^{(12)}(q_{m}-q_n)  \label{ineq:negative-3}\\
	& \quad + \sum_{k=1}^{N_{1}}\sum_{m=1}^{N_{2}}W_{N,-}^{(12)}(p_{k}-q_{m}) \label{ineq:negative-4}\\
	& \quad - N_{1}\sum_{m=1}^{N_{2}}(\rho_{\Phi_{N}}^{(1,0)}\star W_{N,-}^{(12)})(q_{m}) - N_{2}\sum_{k=1}^{N_{1}}(\rho_{\Phi_{N}}^{(0,1)}\star W_{N,-}^{(12)})(p_{k}) \label{ineq:negative-5}\\
	& \quad + N_{1}N_{2}\iint_{\mathbb R^{2} \times \mathbb R^{2}}\rho_{\Phi_{N}}^{(1,0)}(x)\rho_{\Phi_{N}}^{(0,1)}(y)W_{N,+}^{(12)}(x-y){\rm d}x{\rm d}y \nonumber\\
	& \quad  - \frac{N_{1}}{2}W_{N,+}^{(1)}(0) - \frac{N_{2}}{2}W_{N,+}^{(2)}(0) - \frac{N}{2}W_{N,+}^{(12)}(0). \nonumber
	\end{align}
	By using the positivity of $\widehat{W_{N,+}^{(1)}}$, $\widehat{W_{N,+}^{(2)}}$ and $\widehat{W_{N,+}^{(12)}}$, we have
	\begin{align*}
	\eqref{ineq:positive-1}
	& \geq \frac{N_{1}}{2}\iint_{\mathbb R^{2} \times \mathbb R^{2}}\rho_{\Phi_{N}}^{(1,0)}(x)\rho_{\Phi_{N}}^{(1,0)}(y)W_{+}^{(1)}(x-y){\rm d}x{\rm d}y,\\
	\eqref{ineq:positive-2} & \geq \frac{N_{2}}{2}\iint_{\mathbb R^{2} \times \mathbb R^{2}}\rho_{\Phi_{N}}^{(0,1)}(x)\rho_{\Phi_{N}}^{(0,1)}(y)W_{+}^{(2)}(x-y){\rm d}x{\rm d}y.
	\end{align*}
	Next, we apply again Lemma \ref{lem:two-one} for 
		$$
		V^{(1)} = W_{N,-}^{(1)},\quad V^{(2)} = W_{N,-}^{(2)} \quad \text{and} \quad V^{(12)} = W_{N,-}^{(12)}
		$$
		with
		$$
		\chi = (N_{1}-1)\rho_{\Phi_{N}}^{(1,0)} \quad \text{and} \quad \zeta = (N_{2}-1)\rho_{\Phi_{N}}^{(0,1)},
		$$
		we obtain
	\begin{align}
	& \eqref{ineq:negative-1} + \eqref{ineq:negative-2} + \eqref{ineq:negative-3} + \eqref{ineq:negative-4} + \eqref{ineq:negative-5} \nonumber\\
	& \quad  \geq - \frac{(N_{1}-1)^{2}}{2}\iint_{\mathbb R^{2} \times \mathbb R^{2}}\rho_{\Phi_{N}}^{(1,0)}(x)\rho_{\Phi_{N}}^{(1,0)}(y)\big(W_{N,-}^{(1)}+W_{N,-}^{(12)}\big)(x-y){\rm d}x{\rm d}y \label{ineq:negative-6} \\
	& \quad  \quad - \frac{(N_{2}-1)^{2}}{2}\iint_{\mathbb R^{2} \times \mathbb R^{2}}\rho_{\Phi_{N}}^{(0,1)}(x)\rho_{\Phi_{N}}^{(0,1)}(y)\big(W_{N,-}^{(2)}+W_{N,-}^{(12)}\big)(x-y){\rm d}x{\rm d}y \label{ineq:negative-7}\\
	& \quad  \quad - (N_{1}-1)(N_{2}-1)\iint_{\mathbb R^{2} \times \mathbb R^{2}}\rho_{\Phi_{N}}^{(1,0)}(x)\rho_{\Phi_{N}}^{(0,1)}(y)W_{N,-}^{(12)}(x-y){\rm d}x{\rm d}y \nonumber \\
	& \quad  \quad - \sum_{k=1}^{N_{1}}(\rho_{\Phi_{N}}^{(0,1)}\star W_{N,-}^{(12)})(p_{k}) - \sum_{m=1}^{N_{2}}(\rho_{\Phi_{N}}^{(1,0)}\star W_{N,-}^{(12)})(q_{m}) \nonumber\\
	& \quad  \quad - \frac{N_{1}}{2}W_{N,-}^{(1)}(0) - \frac{N_{2}}{2}W_{N,-}^{(2)}(0) - \frac{N}{2}W_{N,-}^{(12)}(0).\nonumber
	\end{align}
	Again, by using the positivity of $\widehat{W_{N,-}^{(1)}}$, $\widehat{W_{N,-}^{(2)}}$ and $\widehat{W_{N,-}^{(12)}}$, we have
	\begin{align*}
	\eqref{ineq:negative-6} & \geq - \frac{N_{1}(N_{1}-1)}{2}\iint_{\mathbb R^{2} \times \mathbb R^{2}}\rho_{\Phi_{N}}^{(1,0)}(x)\rho_{\Phi_{N}}^{(1,0)}(y)\big(W_{N,-}^{(1)}+W_{N,-}^{(12)}\big)(x-y){\rm d}x{\rm d}y \\
	&  = - \frac{N_{1}}{2}\iint_{\mathbb R^{2} \times \mathbb R^{2}}\rho_{\Phi_{N}}^{(1,0)}(x)\rho_{\Phi_{N}}^{(1,0)}(y)W_{-}^{(1)}(x-y){\rm d}x{\rm d}y, \\
	\eqref{ineq:negative-7} & \geq - \frac{N_{2}(N_{2}-1)}{2}\iint_{\mathbb R^{2} \times \mathbb R^{2}}\rho_{\Phi_{N}}^{(0,1)}(x)\rho_{\Phi_{N}}^{(0,1)}(y)\big(W_{N,-}^{(2)}+W_{N,-}^{(12)}\big)(x-y){\rm d}x{\rm d}y \\
	& = - \frac{N_{2}}{2}\iint_{\mathbb R^{2} \times \mathbb R^{2}}\rho_{\Phi_{N}}^{(0,1)}(x)\rho_{\Phi_{N}}^{(0,1)}(y)W_{-}^{(2)}(x-y){\rm d}x{\rm d}y.
	\end{align*}
	Putting all of the above together and using $W_{+}^{(\sigma)}-W_{-}^{(\sigma)}=W^{(\sigma)}$, for $\sigma\in\{1,2,12\}$, we arrive at
	\begin{align*}
	\langle \Phi_{N}|\tilde{H}_{N}|\Phi_{N}\rangle & \geq N_{1}\Big\langle\sqrt{\rho_{\Phi_{N}}^{(1,0)}},h^{(1)}\sqrt{\rho_{\Phi_{N}}^{(1,0)}}\Big\rangle + N_{1}\iint_{\mathbb R^{2} \times \mathbb R^{2}}\rho_{\Phi_{N}}^{(1,0)}(x)\rho_{\Phi_{N}}^{(1,0)}(y)W^{(1)}(x-y){\rm d}x{\rm d}y \\
	& \quad  + N_{2}\Big\langle\sqrt{\rho_{\Phi_{N}}^{(0,1)}},h^{(2)}\sqrt{\rho_{\Phi_{N}}^{(0,1)}}\Big\rangle + N_{2}\iint_{\mathbb R^{2} \times \mathbb R^{2}}\rho_{\Phi_{N}}^{(0,1)}(x)\rho_{\Phi_{N}}^{(0,1)}(y)W^{(2)}(x-y){\rm d}x{\rm d}y \\
	& \quad  + \frac{N_{1}N_{2}}{N}\iint_{\mathbb R^{2} \times \mathbb R^{2}}\rho_{\Phi_{N}}^{(1,0)}(x)\rho_{\Phi_{N}}^{(0,1)}(y)W^{(12)}(x-y){\rm d}x{\rm d}y + \mathcal{R}_{1} \\
	& = N\mathcal{E}^{\rm H}\Big(\sqrt{\rho_{\Phi_{N}}^{(1,0)}},\sqrt{\rho_{\Phi_{N}}^{(0,1)}}\Big) + \mathcal{R}_{1}.
	\end{align*}
	Here, we abbreviated by $\mathcal{R}_{1}$ the error terms
	\begin{align}
	\mathcal{R}_{1} & = - \frac{N_{1}}{2(N_{1}-1)}\big(W_{+}^{(1)}(0)+W_{-}^{(1)}(0)\big) - \frac{N_{2}}{2(N_{2}-1)}\big(W_{+}^{(2)}(0)+W_{-}^{(2)}(0)\big) \nonumber\\
	& \quad - \frac{1}{N}\sum_{k=1}^{N_{1}}(\rho_{\Phi_{N}}^{(0,1)}\star W_{-}^{(12)})(p_{k}) - \frac{1}{N}\sum_{m=1}^{N_{2}}(\rho_{\Phi_{N}}^{(1,0)}\star W_{-}^{(12)})(q_{m}) \label{err:energy-1}\\
	& \quad  + \frac{N-1}{N}\iint_{\mathbb R^{2} \times \mathbb R^{2}}\rho_{\Phi_{N}}^{(1,0)}(x)\rho_{\Phi_{N}}^{(0,1)}(y)W_{-}^{(12)}(x-y){\rm d}x{\rm d}y.\label{err:energy-2}
	\end{align}
	By Young's inequality, we have
	\begin{align*}
	\eqref{err:energy-1}
	& \geq -\frac{N_{1}}{N}\|\rho_{\Phi_{N}}^{(0,1)}\|_{L^1}\|W_{-}^{(12)}\|_{L^\infty} - \frac{N_{2}}{N}\|\rho_{\Phi_{N}}^{(1,0)}\|_{L^1}\|W_{-}^{(12)}\|_{L^\infty} \\
	& = - \|W_{-}^{(12)}\|_{L^\infty} \geq -(2\pi)^{-1}\|\widehat{W_{-}^{(12)}}\|_{L^1} \geq -(2\pi)^{-1}\|\widehat{W^{(12)}}\|_{L^1},\\
	\eqref{err:energy-2} & \geq -\|W_{-}^{(12)}\|_{L^\infty}\|\rho_{\Phi_{N}}^{(1,0)}\|_{L^1}\|\rho_{\Phi_{N}}^{(0,1)}\|_{L^1}  \geq -(2\pi)^{-1}\|\widehat{W^{(12)}}\|_{L^1}.
	\end{align*}
	We note that 
	$$
	W_{+}^{(i)}(0)+W_{-}^{(i)}(0)=(2\pi)^{-1}\|\widehat{W^{(i)}}\|_{L^{1}},
	$$
	for $i\in\{1,2\}$. Hence, we conclude
	$$
	\frac{\langle \Phi_{N}|\tilde{H}_{N}|\Phi_{N}\rangle}{N} \geq E^{\rm H} - \frac{(2\pi)^{-1}}{N}\left( \frac{N_{1}\|\widehat{W^{(1)}}\|_{L^1}}{2(N_{1}-1)} + \frac{N_{2}\|\widehat{W^{(2)}}\|_{L^1}}{2(N_{2}-1)} + 2\|\widehat{W^{(12)}}\|_{L^1}\right).
	$$
	Since the right hand side is independent of the $p_{k}$'s and $q_{m}$'s, the bound
	$$
	\frac{\tilde{H}_{N}}{N} \geq E^{\rm H} - \frac{C}{N}\left( \|\widehat{W^{(1)}}\|_{L^1} + \|\widehat{W^{(2)}}\|_{L^1} + \|\widehat{W^{(12)}}\|_{L^1}\right)
	$$
	holds in the sense of operators in the $2N$-particles space, where $2N = 2N_{1} + 2N_{2}$ with $N_{1}\geq 2$ and $N_{2}\geq 2$. Minimizing over $\Psi_{2N}$ and recalling the upper bound $E_{2N}^{\rm Q} \leq E^{\rm H}$ give the final estimate
	$$
	E^{\rm H} \geq E_{2N}^{\rm Q} \geq E^{\rm H} - \frac{C}{N}\left( \|\widehat{W^{(1)}}\|_{L^1} + \|\widehat{W^{(2)}}\|_{L^1} + \|\widehat{W^{(12)}}\|_{L^1}\right).
	$$
	
	We have considered the case where we have an even number of each type of particles for simplicity. But the proof works the same when one or both of those numbers are odd if we split the system into two groups of $N_{1}+N_{2}$ and $N_{1}+N_{2}+1$, or of $(N_{1}+1)+N_{2}$ and $N_{1}+(N_{2}+1)$. Another possibility is to use the fact that $N \mapsto E_{N}^{\rm Q}/N$ is non-decreasing. We arrive at the final estimate
	$$
	E^{\rm H} \geq E_{N}^{\rm Q} \geq E^{\rm H} - \frac{C}{N}\left( \|\widehat{W^{(1)}}\|_{L^1} + \|\widehat{W^{(2)}}\|_{L^1} + \|\widehat{W^{(12)}}\|_{L^1}\right)
	$$
	for $N = N_{1} + N_{2}$ and $N_{1}\geq 4$ and $N_{2}\geq 4$.
\end{proof}

\subsection{Convergence of the Hartree Energy to the NLS Energy}\label{sec:hartree-nls-energy}

In this subsection, we compare the Hartree and NLS energies. We first note that if $W^{(1)}$, $W^{(2)}$ and $W^{(12)}$ in \eqref{conv:quantum-hartree-energy-general} are replaced by $w_{N}^{(1)}$, $w_{N}^{(2)}$ and $w_{N}^{(12)}$ in \eqref{assumption:potential-1}, then we obtain
\begin{equation}\label{lower:quantum-hartree-energy}
E_{N}^{\rm H} \geq E_{N}^{\rm Q} \geq E_{N}^{\rm H} - CN^{2\beta-1}.
\end{equation}
Here, $E_{N}^{\rm Q}$ and $E_{N}^{\rm H}$ are defined as in \eqref{energy:quantum} and \eqref{energy:hartree-two-component}, respectively. This follows from the fact that the Fourier transform of $w_{N}^{(\sigma)}$ satisfies $\|\widehat{w_{N}^{(\sigma)}}\|_{L^1}\leq CN^{2\beta}$, for $\sigma\in\{1,2,12\}$. Our next step is to estimate the Hartree and NLS energies given by \eqref{energy:hartree-two-component} and \eqref{energy:nls-two-component}. We have the following.

\begin{lemma}\label{lem:energy-hartree-nls}
	Let $0<a_{1},a_{2}<a_{*}$ and either $0<a_{12}<\sqrt{c_{1}^{-1}c_{2}^{-1}(a_{*}-a_{1})(a_{*}-a_{2})}$ or $a_{12}<0$. Then, under assumptions \eqref{assumption:potential-1} and \eqref{assumption:potential-2}, we have
	$$
	\lim_{N\to\infty}E_{N}^{\rm H} = E^{\rm NLS}.
	$$
\end{lemma}

\begin{proof}
	We start with the upper bound. Recall that $\int_{\mathbb R ^{2}}w^{(\sigma)}(x){\rm d}x = 1$, for $\sigma\in\{1,2,12\}$. By introducing the variable $z=N^\beta(x-y)$, we write
	\begin{align}
	& \iint_{\mathbb R^{2}\times \mathbb R^{2}}|u_i(x)|^{2}N^{2\beta}w^{(\sigma)}(N^\beta(x-y))|u_j(y)|^{2}{\rm d}x{\rm d}y - \int_{\mathbb R^{2}}|u_i(x)|^{2}|u_j(x)|^{2}{\rm d}x \nonumber\\
	& \quad  = \iint_{\mathbb R^{2}\times \mathbb R^{2}}|u_i(x)|^{2}w^{(\sigma)}(z)\left(|u_j(x-N^{-\beta}z)|^{2}-|u_j(x)|^{2}\right){\rm d}x{\rm d}z \nonumber\\
	& \quad  = \iint_{\mathbb R^{2}\times \mathbb R^{2}}|u_i(x)|^{2}w^{(\sigma)}(z)\left(\int_0^1 (\nabla|u_j|^{2})(x-tN^{-\beta}z)\cdot(N^{-\beta}z){\rm d}t\right){\rm d}x{\rm d}z, \label{upper-bound:hartree-nls-2}
	\end{align}
	where we made use of the notation
	\begin{equation}\label{notation:sigma}
	\sigma=\begin{cases}1 & \text{if } i=1=j,\\ 2 & \text{if } i=2=j,\\ 12 & \text{if } i=1,j=2 \text{ or } i=2,j=1.\end{cases}
	\end{equation}
	By the diamagnetic inequality $|\nabla |u_j|^{2}|\leq 2|\nabla u_j|\cdot |u_j|$ and H\"{o}lder's inequality, we have
	\begin{equation}\label{upper-bound:hartree-nls-3}
	\int_{\mathbb R^{2}}|u_i(x)|^{2} (\nabla|u_j|^{2})(x-tN^{-\beta}z){\rm d}x \leq 2\|u_i\|_{L^{6}}^{2}\|u_j\|_{L^{6}}\|\nabla u_j\|_{L^{2}}.
	\end{equation}
	From \eqref{upper-bound:hartree-nls-2}, \eqref{upper-bound:hartree-nls-3} and noting that $(1+|z|)w^{(\sigma)}(z)\in L^1(\mathbb R^{2})$, for $\sigma\in\{1,2,12\}$, we obtain
	\begin{equation}\label{upper-bound:hartree-nls-4}
	|\mathcal{E}_{N}^{\rm H}(u_{1},u_{2})-\mathcal{E}^{\rm NLS}(u_{1},u_{2})| \leq N^{-\beta}\mathcal{R}_{2}
	\end{equation}
	where we abbreviated by $\mathcal{R}_{2}$ the error terms
	$$
	\mathcal{R}_{2} = c_{1}a_{1}\|u_{1}\|_{L^{6}}^3\|\nabla u_{1}\|_{L^{2}} + c_{2}a_{2} \|u_{2}\|_{L^{6}}^3\|\nabla u_{2}\|_{L^{2}} + 2c_{1}c_{2}a_{12} \|u_{1}\|_{L^{6}}^{2}\|u_{2}\|_{L^{6}}\|\nabla u_{2}\|_{L^{2}}. 
	$$
	Let $(u_{1},u_{2})$ be a ground state for $E^{\rm NLS}$. It follows from \eqref{upper-bound:hartree-nls-4} that
	\begin{equation}\label{upper:hartree-nls-energy}
	E_{N}^{\rm H} \leq E^{\rm NLS} + N^{-\beta}\mathcal{R}_{2}.
	\end{equation}
	
	Now, we come to the lower bound. Let $(u_{1,N},u_{2,N})$ be a ground state for $E_{N}^{\rm H}$. We first note that, for any $\kappa>0$ and for $\sigma$ given by \eqref{notation:sigma}, we have 
	\begin{equation}\label{ineq:interaction}
	\iint_{\mathbb R^{2} \times \mathbb R^{2}}|u_{i}(x)|^{2}w_{N}^{(\sigma)}(x-y)|u_{j}(y)|^{2}{\rm d}x \leq \frac{\kappa}{2}\int_{\mathbb R^{2}}|u_i(x)|^{4}{\rm d}x + \frac{1}{2\kappa}\int_{\mathbb R^{2}}|u_j(x)|^{4}{\rm d}x.
	\end{equation}
	This follows from the Cauchy--Schwarz inequality and the fact that $\int_{\mathbb R^2}w_{N}^{(\sigma)}(x){\rm d}x = 1$. Applying \eqref{ineq:interaction} several times to $\kappa=1$ with $i=1=j$ and $i=2=j$, we can estimate the intra-species interactions in \eqref{functional:hartree-tow-component} by those in \eqref{functional:nls-two-component} from below. To control the inter-species interactions, we proceed as follows. If $0<a_{1}<a_{*}$, $0<a_{2}<a_{*}$ and $0<a_{12}<\sqrt{c_{1}^{-1}c_{2}^{-1}(a_{*}-a_{1})(a_{*}-a_{2})}$, we can take $\delta \in  \left(\dfrac{c_{1}a_{12}}{a_{*}-a_{2}};\dfrac{a_{*}-a_{1}}{c_{2}a_{12}}\right)$. Then, we use \eqref{ineq:GN-0} and apply \eqref{ineq:interaction} to $\kappa=\delta$ with $i=1$ and $j=2$, we  obtain
	\begin{equation}\label{hartree-nls:inter-specie}
	E_{N}^{\rm H} \geq \frac{c_{1}}{a_{*}}(a_{*}-a_{1}-c_{2}a_{12}\delta) \int_{\mathbb R^2} |\nabla u_{1,N}|^{2}{\rm  d}x + \frac{c_{2}}{a_{*}}(a_{*}-a_{2}-c_{1}a_{12}\delta^{-1}) \int_{\mathbb R^2} |\nabla u_{2,N}|^{2}{\rm  d}x.
	\end{equation}
	In the case of repulsive inter-species interactions, i.e., $a_{12}<0$, we can neglect it and also obtain \eqref{hartree-nls:inter-specie} without the $a_{12}$-term. In any case, \eqref{upper:hartree-nls-energy} and \eqref{hartree-nls:inter-specie} imply that $(u_{1,N},u_{2,N})$ is bounded in $H^{1}(\mathbb R^2) \times H^{1}(\mathbb R^2)$ uniformly in $N$. By Sobolev's inequality, those two  components of the Hartree ground state are bounded in $L^{p}(\mathbb R^2)$ uniformly in $N$, for any $2\leq p <\infty$. Now, one might use \eqref{upper-bound:hartree-nls-2} and \eqref{upper-bound:hartree-nls-3} to compare the inter-species interactions between \eqref{functional:hartree-tow-component} and \eqref{functional:nls-two-component}. We obtain the final estimate
	\begin{equation}\label{lower:hartree-nls-energy}
	E_{N}^{\rm H} \geq E^{\rm NLS} - N^{-\beta}\mathcal{R}_{3}
	\end{equation}
	where we abbreviated by $\mathcal{R}_{3}$ the error terms
	$$
	\mathcal{R}_{3} = 2c_{1}c_{2}a_{12} \|u_{1,N}\|_{L^{6}}^{2}\|u_{2,N}\|_{L^{6}}\|\nabla u_{2,N}\|_{L^{2}}. 
	$$
	
	The convergence of the Hartree energy to the NLS energy follows from \eqref{upper:hartree-nls-energy} and \eqref{lower:hartree-nls-energy}. This  completes the proof of Lemma \ref{lem:energy-hartree-nls}.
\end{proof}

\section{Collapse of the Many-Body System}\label{sec:blow-up-gp}

Through this section, we assume that the potentials $V_{1}$ and $V_{2}$ are of the forms \eqref{potentials:external}, i.e.,
$$
V_{i}(x) = |x-z_{i}|^{p_{i}}, \quad i\in\{1,2\},
$$
where $z_{1},z_{2} \in \mathbb R^{2}$ and $p_{1},p_{2} > 0$. 

\subsection{Proofs of Theorems \ref{thm:blow-up-bec-1-2-I} and \ref{thm:blow-up-bec-12-I}} \label{subsec:proof-I}

The purpose of this subsection is to prove Theorems \ref{thm:blow-up-bec-1-2-I} and \ref{thm:blow-up-bec-12-I} which give the blow-up profile for the many-body system \eqref{hamiltonian} when the total interaction strength of intra-species and inter-species tends to a critical number. We first revisit the blow-up phenomenon for the NLS minimization problem \eqref{energy:nls-two-component}. In the case $a_{12}>0$, the existence of the NLS ground states follows the standard direct method in the calculus of variations. The following is taken from \cite[Theorem 1.1]{GuoZenZho-17-dcds}, but the statement is adapted to our model \eqref{energy:nls-two-component}.

\begin{theorem}\label{thm:existence-nls-ground-states-I} 
	We have the followings
	\begin{itemize}
		\item[(i)] If $0<a_{1},a_{2}<a_{*}$ and $0<a_{12}<\sqrt{c_{1}^{-1}c_{2}^{-1}(a_{*}-a_{1})(a_{*}-a_{2})}$, then $E^{\rm NLS}\geq 0$ and it has at least one ground state.
		\item[(ii)] If either $a_{1}>a_{*}$ or $a_{2}>a_{*}$ or $a_{12}>2^{-1}c_{1}^{-1}c_{2}^{-1}(a_{*}-c_{1}a_{1}-c_{2}a_{2})$, then $E^{\rm NLS}=-\infty$.
	\end{itemize}
\end{theorem}

As pointed out in \cite[Theorem 1.2 and Theorem 1.3]{GuoZenZho-17-dcds}, when $0<a_{1},a_{2}<a_{*}$ and 
$$
\sqrt{c_{1}^{-1}c_{2}^{-1}(a_{*}-a_{1})(a_{*}-a_{2})} \leq a_{12} \leq 2^{-1}c_{1}^{-1}c_{2}^{-1}(a_{*}-c_{1}a_{1}-c_{2}a_{2}),
$$
then there may exist ground states for \eqref{energy:nls-two-component}, under additional assumptions on $(a_{1},a_{2},a_{12})$, especially when $z_{1} \ne z_{2}$. Therefore, it is reasonable to consider the case $z_{1} \equiv z_{2}$ in order to study the limit behavior of ground states when they do not exist at the threshold. The following is taken from \cite[Theorem 1.5]{GuoZenZho-17-dcds}.

\begin{theorem} \label{thm:blow-up-nls-ground-states-I} 
	Assume that $0<a_{12}<a_{*}\min\{c_{1}^{-1},c_{2}^{-1}\}$ is fixed and $V_{1}$, $V_{2}$ are defined as in \eqref{potentials:external} with $z_{1} = 0 = z_{2}$. Then, for every sequence $(a_{1,N},a_{2,N})\nearrow (a_{*}-c_{2}a_{12},a_{*}-c_{1}a_{12})$ as $N\to\infty$, we have
	\begin{equation}\label{behavior:nls-energy-1-2-I}
	E^{\rm NLS} = (a_{*}-a_{N})^{\frac{p_{0}}{p_{0}+2}} \left(\frac{p_{0}+2}{p_{0}} \cdot \frac{\Lambda^{2}}{a_{*}}+o(1)_{N\to\infty}\right),
	\end{equation}
	where $a_{N} = c_{1}a_{1,N} + c_{2}a_{2,N} + 2c_{1}c_{2}a_{12}$, $p_{0} = \min\{p_{1},p_{2}\}$ and $\Lambda$ is given by \eqref{lambda:1-2-I}.
	
	In addition, assume that $(u_{1,N},u_{2,N})$ is a positive ground state for $E^{\rm NLS}$ in \eqref{energy:nls-two-component} for each $0<a_{1,N}<a_{*}-c_{2}a_{12}$ and $0<a_{2,N}<a_{*}-c_{1}a_{12}$. Then, up to extraction of a subsequence, we have
	$$
	\lim_{N\to\infty} \ell_{N}^{-1} u_{1,N}(\ell_{N}^{-1}\cdot) = Q_{0} = \lim_{N\to\infty} \ell_{N}^{-1}u_{2,N}(\ell_{N}^{-1}\cdot)
	$$
	strongly in $H^{1}(\mathbb R^2)$, where $\ell_{N} = \Lambda(a_{*}-a_{N})^{-\frac{1}{p_{0}+2}}$ and $Q_{0}$ is given by \eqref{GN:normalized}.
	
\end{theorem}

Theorem \ref{thm:blow-up-nls-ground-states-I}  gives the blow-up profile for the NLS \emph{exact} ground states when they exist. The proof in \cite{GuoZenZho-17-dcds} is based on a detailed analysis of the Euler--Lagrange equation associated with the ground states. In order to establish the blow-up behavior of the many-body ground states via the Feynman--Hellmann-type argument, we need to extend that blow-up result to the NLS \emph{approximate} ground states. This is done in the one-component setting \cite{LewNamRou-17-proc}. However, the situation is more complicated in the two-component case, due to the presence of the (attractive) inter-species interactions. In fact, we will need to consider the blow-up behavior of the following (modified) Hartree variational problem
\begin{equation}\label{energy:modified-hartree-two-component}
E_{N}^{\rm mH} := \inf_{\substack{u_{1},u_{2}\in H^{1}(\mathbb R^{2})\\ \|u_{1}\|_{L^{2}} = 1 = \|u_{2}\|_{L^{2}}}}\mathcal{E}_{N}^{\rm mH}(u_{1},u_{2}),
\end{equation}
where the (modified) Hartree functional is given by
\begin{align}
\mathcal{E}_{N}^{\rm mH} (u_{1},u_{2}) & = c_{1}\int_{\mathbb R^{2}}\left[|\nabla u_{1}(x)|^{2}+V_{1}(x)|u_{1}(x)|^{2}-\frac{a_{1}}{2}|u_{1}(x)|^{4}\right]{\rm d}x \nonumber \\
& \quad  + c_{2}\int_{\mathbb R^{2}}\left[|\nabla u_{2}(x)|^{2}+V_{2}(x)|u_{2}(x)|^{2}-\frac{a_{2}}{2}|u_{2}(x)|^{4}\right]{\rm d}x \nonumber \\
& \quad  - c_{1}c_{2}a_{12}\int_{\mathbb R^{2}}|u_{1}(x)|^{2}(w_{N}^{(12)}\star|u_{2}|^{2})(x){\rm d}x. \label{functional:modified-hartree}
\end{align}
One can see that \eqref{functional:modified-hartree} interpolates between Hartree \eqref{functional:hartree-tow-component} and NLS \eqref{functional:nls-two-component}. The asymptotic formula for the (modified) Hartree energy will be given in the following.

\begin{lemma}\label{lem:modified-hartree-energy-1-2-I}
	Assume that $0<a_{12}<\min\{c_{1}^{-1},c_{2}^{-1}\}$ is fixed and $V_{1}$, $V_{2}$ are defined as in \eqref{potentials:external} with $z_{1} = 0 = z_{2}$. Let $(a_{1},a_{2}) := (a_{1,N},a_{2,N})\nearrow (a_{*}-c_{2}a_{12},a_{*}-c_{1}a_{12})$ such that $a_{N} := c_{1}a_{1,N} + c_{2}a_{2,N} + 2c_{1}c_{2}a_{12} = a_{*}-N^{-\gamma}$ with
	$$
	0<\gamma < \frac{p_{0}+2}{p_{0}+3}\beta ,\quad p_{0}=\min\{p_{1},p_{2}\}.
	$$
	Then, we have, with $\Lambda$ given by \eqref{lambda:1-2-I},
	\begin{equation}\label{asymptotic:modifiled-hartree-energy-1-2-I}
	E_{N}^{\rm mH} = E^{\rm NLS} + o(E^{\rm NLS})_{N\to\infty} = (a_{*}-a_{N})^{\frac{p_{0}}{p_{0}+2}} \left(\frac{p_{0}+2}{p_{0}} \cdot \frac{\Lambda^{2}}{a_{*}}+o(1)_{N\to\infty}\right).
	\end{equation}
\end{lemma}

\begin{proof}
	We start with the upper bound. By the arguments in the proof of Lemma \ref{lem:energy-hartree-nls} and the asymptotic behavior of the NLS ground states (see, e.g., \cite[Proposition 3]{GuoZenZho-17-dcds}), we have
	\begin{equation}\label{upper-bound-hartree-1-2-I}
	E_{N}^{\rm mH} \leq  E^{\rm NLS} + CN^{-\beta}\ell_{N}^3  = (a_{*}-a_{N})^{\frac{p_{0}}{p_{0}+2}} \left[\frac{p_{0}+2}{p_{0}} \cdot \frac{\Lambda^{2}}{a_{*}} + CN^{-\beta}(a_{*}-a_{N})^{-\frac{p_{0}+3}{p_{0}+2}}\right].
	\end{equation}
	The error term $N^{-\beta}(a_{*}-a_{N})^{-\frac{p_{0}+3}{p_{0}+2}}$ is of order $1$ when $a_{*}-a_{N}=N^{-\gamma}$ with $0<\gamma<\frac{p_{0}+2}{p_{0}+3}\beta$.
	
	Now, we turn to the lower bound. Let $(u_{1,N},u_{2,N})\in H^1(\mathbb R^{2})\times H^1(\mathbb R^{2})$ be a ground state for $E_{N}^{\rm mH}$ in \eqref{functional:hartree-tow-component}, which exists under the conditions $0<a_{1}<a_{*}$ and $0<a_{2}<a_{*}$. Let us rewrite $\mathcal{E}_{N}^{\rm mH}$ as follows
	\begin{align}
	\mathcal{E}_{N}^{\rm mH} (u_{1,N},u_{2,N}) 	& = c_{1}\int_{\mathbb R^{2}}\left[|\nabla u_{1,N}(x)|^{2} - \frac{\tilde{a}_{1,N}}{2}|u_{1,N}(x)|^{4}\right]{\rm d}x + c_{1}\int_{\mathbb R^{2}}|x|^{p_{1}}|u_{1,N}(x)|^{2}{\rm d}x \nonumber\\
	& \quad  + c_{2}\int_{\mathbb R^{2}}\left[|\nabla u_{2,N}(x)|^{2} - \frac{\tilde{a}_{2,N}}{2}|u_{2,N}(x)|^{4}\right]{\rm d}x + c_{2}\int_{\mathbb R^{2}}|x|^{p_{2}}|u_{2,N}(x)|^{2}{\rm d}x \nonumber\\
	& \quad  + \frac{c_{1}c_{2}a_{12}}{2}\int_{\mathbb R^{2}}\left[|u_{1,N}(x)|^{4} + |u_{2,N}(x)|^{4} - 2|u_{1}(x)|^{2}(w_{N}^{(12)}\star|u_{2}|^{2})(x)\right]{\rm d}x \label{approx-nls-ground-state-1-2-I-1}
	\end{align}
	where $\tilde{a}_{1,N}=a_{1,N}+c_{2}a_{12}$ and $\tilde{a}_{2,N}=a_{2,N}+c_{1}a_{12}$. Applying \eqref{ineq:interaction} to $\kappa=1$, $i=1$, $j=2$, we obtain
	\begin{equation}\label{lower-bound-1-2-I-3}
	E_{N}^{\rm mH} \geq c_{1}\int_{\mathbb R^{2}}\left(|\nabla u_{1,N}(x)|^{2} - \frac{\tilde{a}_{1,N}}{2}|u_{1,N}(x)|^{4}\right){\rm d}x + c_{2}\int_{\mathbb R^{2}}\left(|\nabla u_{2,N}(x)|^{2} - \frac{\tilde{a}_{2,N}}{2}|u_{2,N}(x)|^{4}\right){\rm d}x.
	\end{equation}
	For each $N$, we may assume without loss of generality that $\tilde{a}_{1,N}\leq \tilde{a}_{2,N}$. By \eqref{lower-bound-1-2-I-3} and \eqref{ineq:GN-0}, we have
	\begin{equation}\label{lower-bound-1-2-I-4}
	E_{N}^{\rm mH} \geq c_{1}\frac{a_{*}-\tilde{a}_{1,N}}{a_{*}}\int_{\mathbb R^{2}}|\nabla u_{1,N}(x)|^{2} {\rm d}x \geq c_{1}\frac{a_{*}-a_{N}}{a_{*}}\int_{\mathbb R^{2}}|\nabla u_{1,N}(x)|^{2} {\rm d}x,
	\end{equation}
	where $a_{N}=c_{1}\tilde{a}_{1,N}+c_{2}\tilde{a}_{2,N}$. We denote $\tilde{u}_{1,N}=\ell_{N}^{-1}u_{1,N}(\ell_{N}^{-1}\cdot)$ and $\tilde{u}_{2,N}=\ell_{N}^{-1}u_{2,N}(\ell_{N}^{-1}\cdot)$ where $\ell_{N}$ is defined as in Theorem \ref{thm:blow-up-nls-ground-states-I}. Then, it follows from \eqref{lower-bound-1-2-I-4} and the upper bound of $E_{N}^{\rm mH}$ in \eqref{upper-bound-hartree-1-2-I} that $\{\tilde{u}_{1,N}\}$ is bounded in $H^1(\mathbb{R}^{2})$. Next, we prove that $\{\tilde{u}_{2,N}\}$ is also bounded in $H^1(\mathbb{R}^{2})$. We notice that, by again \eqref{lower-bound-1-2-I-3} and \eqref{ineq:GN-0}, we have
	\begin{equation}\label{bounded-1-2-I-2}
	\int_{\mathbb R^{2}}|\nabla \tilde{u}_{2,N}(x)|^{2} {\rm d}x - \frac{\tilde{a}_{2,N}}{2}\int_{\mathbb{R}^{2}}|\tilde{u}_{2,N}(x)|^{4}{\rm d}x \leq c_{2}^{-1}\ell_{N}^{-2}E_{N}^{\rm mH}.
	\end{equation}
	On the other hand, by applying \eqref{ineq:interaction} to $\kappa=2$, $i=1$, $j=2$ and using \eqref{ineq:GN-0} we obtain
	$$
	E_{N}^{\rm mH} \geq \frac{c_{1}c_{2}a_{12}}{2}\left(-\int_{\mathbb R^{2}}|u_{1,N}(x)|^{4}{\rm d}x + \frac{1}{2}\int_{\mathbb R^{2}}|u_{2,N}(x)|^{4}{\rm d}x\right)
	$$
	which implies that
	\begin{equation}\label{bounded-1-2-I-3}
	-\int_{\mathbb R^{2}}|\tilde{u}_{1,N}(x)|^{4}{\rm d}x + \frac{1}{2}\int_{\mathbb R^{2}}|\tilde{u}_{2,N}(x)|^{4}{\rm d}x \leq 2c_{1}^{-1}c_{2}^{-1}a_{12}^{-1}\ell_{N}^{-2}E_{N}^{\rm mH}.
	\end{equation}
	Hence, we deduce from \eqref{bounded-1-2-I-2}, \eqref{bounded-1-2-I-3} and the upper bound of $E_{N}^{\rm mH}$ in \eqref{upper-bound-hartree-1-2-I} that $\{\tilde{u}_{2,N}\}$ is bounded in $H^1(\mathbb{R}^{2})$. Since this holds for each $N$, we conclude that the boundedness of $\{\tilde{u}_{1,N}\}$ and $\{\tilde{u}_{2,N}\}$ holds for the whole sequence. We may apply \eqref{upper-bound:hartree-nls-2} and \eqref{upper-bound:hartree-nls-3} with $i=1$, $j=2$ to obtain
	$$
	E_{N}^{\rm mH} \geq E^{\rm NLS} - CN^{-\beta}\ell_{N}^3  = (a_{*}-a_{N})^{\frac{p_{0}}{p_{0}+2}} \left[\frac{p_{0}+2}{p_{0}} \cdot \frac{\Lambda^{2}}{a_{*}} + CN^{-\beta}(a_{*}-a_{N})^{-\frac{p_{0}+3}{p_{0}+2}}\right].
	$$
	The error term in the above matches the one in \eqref{upper-bound-hartree-1-2-I}, up to a constant. This concludes the proof of Lemma \ref{lem:modified-hartree-energy-1-2-I}.
\end{proof}

\begin{remark}\label{rem:asymptotic-quantum-energy}
	\begin{itemize}
		\item By the arguments in the proof of Lemmas \ref{lem:energy-hartree-nls} and \ref{lem:modified-hartree-energy-1-2-I}, we obtain the same asymptotic formula as in \eqref{asymptotic:modifiled-hartree-energy-1-2-I} for the Hartree energy \eqref{energy:hartree-two-component}.
		\item It follows from \eqref{lower:quantum-hartree-energy} and the arguments in the proof of Lemma \ref{lem:modified-hartree-energy-1-2-I} that
		\begin{equation}\label{upper:quantum-modified-hartree}
		E_{N}^{\rm mH} + CN^{-\beta}\ell_{N}^{3} \geq E_{N}^{\rm H} \geq E_{N}^{\rm Q} \geq E_{N}^{\rm H} - CN^{2\beta-1} \geq E_{N}^{\rm mH} - CN^{2\beta-1}.
		\end{equation}
		Then, the asymptotic formula of the (modified) Hartree energy implies that of the quantum energy, which gives \eqref{asymptotic:quantum-energy} in Theorem \ref{thm:blow-up-bec-1-2-I}. Since $CN^{2\beta-1}$ is the error in the above energy estimate, the following conditions are taken into account
		$$
		0<\beta<1/2 \quad \text{and} \quad 0<\gamma < \frac{p_{0}+2}{p_{0}}(1-2\beta) \quad \text{with} \quad p_{0}=\min\{p_{1},p_{2}\}.
		$$
	\end{itemize}
\end{remark}

Having the blow-up behavior of the (modified) Hartree energy, we are now able to study that of its (approximate) ground states. We have the following.

\begin{theorem}\label{thm:blowup-modified-hartree-approximate-1-2-I}
	Assume that $0<a_{12}<a_{*}\min\{c_{1}^{-1},c_{2}^{-1}\}$ is fixed and $V_{1}$, $V_{2}$ are defined as in \eqref{potentials:external} with $z_{1} = 0 = z_{2}$. Let $(a_{1,N},a_{2,N})\nearrow (a_{*}-c_{2}a_{12},a_{*}-c_{1}a_{12})$ as $N\to\infty$ such that $a_{N} := c_{1}a_{1,N} + c_{2}a_{2,N} + 2c_{1}c_{2}a_{12} = a_{*}-N^{-\gamma}$ with
	$$
	0<\gamma < \frac{p_{0}+2}{p_{0}+3}\beta ,\quad p_{0}=\min\{p_{1},p_{2}\}.
	$$
	Let $(u_{1,N},u_{2,N}) \in H^{1}(\mathbb{R}^{2})\times H^{1}(\mathbb{R}^{2})$ be a sequence of a coupled positive functions such that $\|u_{1,N}\|_{L^{2}}=1=\|u_{2,N}\|_{L^{2}}$ and
	\begin{equation}\label{behavior:approximate-modified-hartree-energy-1-2-I}
	\mathcal{E}_{N}^{\rm mH} (u_{1,N},u_{2,N}) = E_{N}^{\rm mH} + o(E_{N}^{\rm mH})_{N\to\infty} = (a_{*}-a_{N})^{\frac{p_{0}}{p_{0}+2}} \left(\frac{p_{0}+2}{p_{0}} \cdot \frac{\Lambda^{2}}{a_{*}}+o(1)_{N\to\infty}\right),
	\end{equation}
	where $\Lambda$ is given by \eqref{lambda:1-2-I}. Then, we have
	\begin{equation}\label{behavior:approximate-modified-hartree-ground-state-1-2-I}
	\lim_{N\to\infty} \ell_{N}^{-1} u_{1,N}(\ell_{N}^{-1}\cdot) = Q_{0} = \lim_{N\to\infty} \ell_{N}^{-1}u_{2,N}(\ell_{N}^{-1}\cdot)
	\end{equation}
	strongly in $L^{2}(\mathbb R^2)$, where $\ell_{N}$ is defined as in Theorem \ref{thm:blow-up-nls-ground-states-I} and $Q_{0}$ is given by \eqref{GN:normalized}. 
\end{theorem}

\begin{proof}
	Denote $\tilde{u}_{1,N}=\ell_{N}^{-1}u_{1,N}(\ell_{N}^{-1}\cdot)$ and $\tilde{u}_{2,N}=\ell_{N}^{-1}u_{2,N}(\ell_{N}^{-1}\cdot)$; then, $\|\tilde{u}_{1,N}\|_{L^{2}} = 1 = \|\tilde{u}_{2,N}\|_{L^{2}}$. By the same argument as in the proof of Lemma \ref{lem:modified-hartree-energy-1-2-I}, we can prove that both sequences $\{\tilde{u}_{1,N}\}$ and $\{\tilde{u}_{2,N}\}$ are bounded in $H^{1}(\mathbb R^2)$. Thus, $\tilde{u}_{1,N}$ (resp. $\tilde{u}_{2,N}$) converges to a function $W_{1}$ (resp. $W_{2}$) weakly in $H^1(\mathbb{R}^{2})$ and pointwise almost everywhere in $\mathbb{R}^{2}$. We will show that $W_{1} \equiv W_{2}$. Indeed, by applying \eqref{upper-bound:hartree-nls-2}, \eqref{upper-bound:hartree-nls-3} with $i=1$, $j=2$ and using \eqref{ineq:GN-0} we obtain
	$$
	\mathcal{E}_{N}^{\rm mH} (u_{1,N},u_{2,N}) \geq \frac{c_{1}c_{2}a_{12}}{2}\int_{\mathbb R^{2}}\left[|u_{1,N}(x)|^{2} - |u_{2,N}(x)|^{2}\right]^{2}{\rm d}x - CN^{-\beta}\ell_{N}^3,
	$$
	which in turn implies that
	\begin{equation}\label{approx-nls-ground-state-1-2-I-3}
	\||\tilde{u}_{1,N}|^{2}-|\tilde{u}_{2,N}|^{2}\|_{L^{2}}^{2} \leq 2c_{1}^{-1}c_{2}^{-1}a_{12}^{-1}\left(\ell_{N}^{-2}\mathcal{E}_{N}^{\rm mH} (u_{1,N},u_{2,N}) + CN^{-\beta}\ell_{N}\right).
	\end{equation}
	By taking the limit $N\to\infty$ in \eqref{approx-nls-ground-state-1-2-I-3} and using \eqref{behavior:approximate-modified-hartree-energy-1-2-I}, we conclude that $W_{1} = W_{0} = W_{2}$ almost everywhere in $\mathbb R^{2}$.
	
	On the other hand, since $p_{0} = \min\{p_{1},p_{2}\}$, we deduce from \eqref{approx-nls-ground-state-1-2-I-1}, \eqref{behavior:approximate-modified-hartree-energy-1-2-I} that either $\int_{\mathbb R^{2}} |x|^{p_{1}}|\tilde{u}_{1,N}(x)|^{2}{\rm d}x$ or $\int_{\mathbb R^{2}} |x|^{p_{2}}|\tilde{u}_{2,N}(x)|^{2}{\rm d}x$ is bounded. It then follows that either $\tilde{u}_{1,N}$ or $\tilde{u}_{2,N}$ converges to $W_{0}$ strongly in $L^{r}(\mathbb{R}^{2})$, for $2\leq r<\infty$. In particular, we have $\|W_{0}\|_{L^{2}}=1$. Moreover, by the Cauchy--Schwarz inequality and Minkowski's inequality we have
	\begin{align*}
	|\|\tilde{u}_{1,N}\|_{L^{4}}^{4} - \|\tilde{u}_{2,N}\|_{L^{4}}^{4}| & \leq \||\tilde{u}_{1,N}|^{2}+|\tilde{u}_{2,N}|^{2}\|_{L^{2}} \||\tilde{u}_{1,N}|^{2}-|\tilde{u}_{2,N}|^{2}\|_{L^{2}} \\
	& \leq (\|\tilde{u}_{1,N}\|_{L^{4}}^{2} + \|\tilde{u}_{2,N}\|_{L^{4}}^{2}) \||\tilde{u}_{1,N}|^{2}-|\tilde{u}_{2,N}|^{2}\|_{L^{2}},
	\end{align*}
	which implies that
	\begin{equation}\label{approx-nls-ground-state-1-2-I-4}
	|\|\tilde{u}_{1,N}\|_{L^{4}}^{2} - \|\tilde{u}_{2,N}\|_{L^{4}}^{2}| \leq \||\tilde{u}_{1,N}|^{2}-|\tilde{u}_{2,N}|^{2}\|_{L^{2}}.
	\end{equation}
	It follows from \eqref{behavior:approximate-modified-hartree-energy-1-2-I}, \eqref{approx-nls-ground-state-1-2-I-3} and \eqref{approx-nls-ground-state-1-2-I-4} that both $\tilde{u}_{1,N}$ and $\tilde{u}_{2,N}$ converge to $W_{0}$ strongly in $L^{4}(\mathbb{R}^{2})$. In fact, those convergences hold in $L^{r}(\mathbb{R}^{2})$, for $4\leq r<\infty$, by the $H^1(\mathbb{R}^{2})$-boundedness of $\{\tilde{u}_{1,N}\}$ and $\{\tilde{u}_{2,N}\}$. Taking the limit $N\to\infty$ in \eqref{bounded-1-2-I-2} and using  Fatou's lemma and the Hardy--Littlewood--Sobolev inequality, we obtain
	$$
	\int_{\mathbb R^{2}}|\nabla W_{0}(x)|^{2} {\rm d}x - \frac{a_{*}}{2}\int_{\mathbb{R}^{2}}|W_{0}(x)|^{4}{\rm d}x \leq 0.
	$$
	Thus, $W_{0}$ is an optimizer for \eqref{ineq:GN-0}. Recall that \eqref{ineq:GN-0} admits a unique optimizer, up to translation and dilations. Therefore, a simple scaling and the uniqueness (up to translation) of positive solutions of \eqref{eq:GN} allows us to conclude that
	$$
	W_{0}(x) = (a_{*})^{-\frac{1}{2}}bQ(bx+x_{0})
	$$
	for some constant $b\in\mathbb R^+$ and $x_{0}\in\mathbb R^{2}$. Here, $Q$ is the unique (up to translation) solution of \eqref{eq:GN}. We will show that\ $b=1$ and $x_{0}=0$. Indeed, it follows from \eqref{behavior:approximate-modified-hartree-energy-1-2-I}, \eqref{approx-nls-ground-state-1-2-I-1}, \eqref{ineq:GN-0} and Fatou's lemma that
	\begin{equation}\label{approx-nls-ground-state-1-2-I-6}
	\frac{p_{0}+2}{p_{0}}\Lambda^{2} \geq \frac{b^{2}\Lambda^{2}}{a_{*}}\int_{\mathbb R^{2}}|\nabla Q(x)|^{2} {\rm d}x + \frac{\nu}{b^{p_{0}}\Lambda^{p_{0}}}\int_{\mathbb{R}^{2}}|x|^{p_{0}}|Q(x+b^{-1}x_{0})|^{2}{\rm d}x.
	\end{equation}
	Here, we have used the fact that $c_{1}+c_{2}=1$ and the assumption $p_{0}=\min\{p_{1},p_{2}\}$. Note that $\|\nabla Q\|_{L^{2}}^{2} = \|Q\|_{L^{2}}^{2} = a_{*}$ and
	\begin{equation}\label{approx-nls-ground-state-1-2-I-7}
	\int_{\mathbb{R}^{2}}|x|^{p_{0}}|Q(x+b^{-1}x_{0})|^{2}{\rm d}x \geq \int_{\mathbb{R}^{2}}|x|^{p_{0}}|Q(x)|^{2}{\rm d}x,
	\end{equation}
	by the Hardy--Littlewood rearrangement inequality as $Q$ is a radial symmetric decreasing function. Thus, \eqref{approx-nls-ground-state-1-2-I-6} reduces to
	\begin{equation}\label{approx-nls-ground-state-1-2-I-8}
	\frac{p_{0}+2}{p_{0}}\Lambda^{2} \geq b^{2}\Lambda^{2} + \frac{\nu}{b^{p_{0}}\Lambda^{p_{0}}}\int_{\mathbb{R}^{2}}|x|^{p_{0}}|Q(x)|^{2}{\rm d}x.
	\end{equation}
	It is elementary to check that
	$$
	\inf_{\lambda>0}\left(\lambda^{2} + \frac{\nu}{\lambda^{p_{0}}}\int_{\mathbb{R}^{2}}|x|^{p_{0}}|Q(x)|^{2}{\rm d}x\right) = \frac{p_{0}+2}{p_{0}}\Lambda^{2}
	$$
	with the unique optimal value $\lambda=\Lambda$. Therefore, the equality in \eqref{approx-nls-ground-state-1-2-I-8} must occurs, and hence, $b=1$. This also implies that the equality in \eqref{approx-nls-ground-state-1-2-I-7} must occurs, and hence, $x_{0}=0$.
\end{proof}

\begin{remark} By the arguments in the proof of Theorem \ref{thm:blowup-modified-hartree-approximate-1-2-I}, we obtain the same blow-up behavior as in \eqref{behavior:approximate-modified-hartree-ground-state-1-2-I} for the Hartree and NLS (approximate) ground states.
	
\end{remark}

In the case of the totally attractive system, we have seen in Theorems \ref{thm:blow-up-nls-ground-states-I}  and \ref{thm:blowup-modified-hartree-approximate-1-2-I} that the two components of NLS exact/approximate ground states have the same behavior when $a_{12}>0$ is fixed and $(a_{1},a_{2}) \nearrow (a_{*}-c_{2}a_{12},a_{*}-c_{1}a_{12})$. It is the same situation when $0<a_{1},a_{2}<a_{*}$ are fixed and $a_{12}>0$ tends to a critical number. In that case, the limit behavior of the NLS \emph{exact} ground states has not been studied in \cite{GuoZenZho-17-dcds}. But it is somehow similar to the previous case. As mentioned in the introduction, we will assume that $c_{1}(a_{*}-a_{1}) = c_{1}c_{2}\alpha_{*} = c_{2}(a_{*}-a_{2})$. Note that, with this assumption, Theorem \ref{thm:existence-nls-ground-states-I} gives a complete classification of the existence and non-existence of ground states for \eqref{energy:nls-two-component}. In the following, we address the limiting profile of the general NLS \emph{approximate} ground states in the limit regime that the inter-species interactions tend to the critical number $\alpha_{*}$. We first note that we have the following estimate
\begin{equation}\label{behavior:nls-energy-12-I-upper}
\limsup_{a_{12}\nearrow\alpha_{*}}\frac{E^{\rm NLS}}{(\alpha_{*}-a_{12})^{\frac{p_{0}}{p_{0}+2}}} \leq 2c_{1}c_{2}\frac{p_{0}+2}{p_{0}} \cdot \frac{\Theta^{2}}{a_{*}},
\end{equation}
where $p_{0}=\min\{p_{1},p_{2}\}$ and $\Theta$ is given by \eqref{lambda:12-I}. To see this, we simply take
$$
u_{1}(x) = u_{2}(x)= (a_{*})^{\frac{1}{2}}\lambda(\alpha_{*}-a_{12})^{-\frac{1}{p_{0}+2}}Q(\lambda(\alpha_{*}-a_{12})^{-\frac{1}{p_{0}+2}}x)
$$
as a trial function for $E^{\rm NLS}$ in \eqref{energy:nls-two-component} and minimize it over $\lambda > 0$.

The estimate \eqref{behavior:nls-energy-12-I-upper} suggested the asymptotic formula of the NLS energy in the limit regime $0<a_{12} := \alpha_{N} \nearrow \alpha_{*}$ as $N\to\infty$. This will imply the asymptotic formula of the (modified) Hartree energy \eqref{energy:modified-hartree-two-component}. More precisely, we have
\begin{equation}\label{asymptotic:modifiled-hartree-energy-12-I}
E_{N}^{\rm mH} = E^{\rm NLS} + o(E^{\rm NLS})_{N\to\infty} = (\alpha_{*}-\alpha_{N})^{\frac{p_{0}}{p_{0}+2}} \left(2c_{1}c_{2}\frac{p_{0}+2}{p_{0}} \cdot \frac{\Theta^{2}}{a_{*}}+o(1)_{N\to\infty}\right),
\end{equation}
provided that 
$$
0<\beta<1/2 \quad \text{and} \quad \alpha_{N} = \alpha_{*} - N^{-\gamma} \quad \text{with} \quad 0 < \gamma < \frac{p_{0}+2}{p_{0}+3}\beta.
$$
In addition, if $\gamma < \frac{p_{0}+2}{p_{0}}(1-2\beta)$, then we also obtain the same asymptotic formula as in \eqref{asymptotic:modifiled-hartree-energy-12-I} for the quantum energy. Furthermore, the blow-up behavior of the (modified) Hartree and the NLS (approximate) ground states $(u_{1,N},u_{2,N})$ is also obtained and we have
\begin{equation}\label{behavior:approximate-modified-hartree-ground state-12-I}
\lim_{N\to\infty} \ell_{N}^{-1} u_{1,N}(\ell_{N}^{-1}\cdot) = Q_{0} = \lim_{N\to\infty} \ell_{N}^{-1}u_{2,N}(\ell_{N}^{-1}\cdot)
\end{equation}
strongly in $L^{2}(\mathbb R^2)$, where $\ell_{N} = \Theta(\alpha_{*}-\alpha_{N})^{-\frac{1}{p_{0}+2}}$ and $Q_{0}$ is given by \eqref{GN:normalized}. The proofs of \eqref{asymptotic:modifiled-hartree-energy-12-I} and \eqref{behavior:approximate-modified-hartree-ground state-12-I} follow the same (even simpler) from that of \eqref{asymptotic:modifiled-hartree-energy-1-2-I} and \eqref{behavior:approximate-modified-hartree-ground-state-1-2-I}. Here, we omit the details for brevity. Moreover, the following estimate was derived in the energy estimate
\begin{equation}\label{behavior:nls-energy-12-I-lower}
\liminf_{a_{12}\nearrow\alpha_{*}}\frac{E^{\rm NLS}}{(\alpha_{*}-a_{12})^{\frac{p_{0}}{p_{0}+2}}} \geq 2c_{1}c_{2}\frac{p_{0}+2}{p_{0}} \cdot \frac{\Theta^{2}}{a_{*}},
\end{equation}
which together with \eqref{behavior:nls-energy-12-I-upper} yields the asymptotic behavior of the NLS energy.

Now, we are in the position to give the proofs of Theorems \ref{thm:blow-up-bec-1-2-I} and \ref{thm:blow-up-bec-12-I}. We only prove Theorem \ref{thm:blow-up-bec-1-2-I} since the proof of Theorem \ref{thm:blow-up-bec-12-I} is analogous.

\emph{Proof of Theorem \ref{thm:blow-up-bec-1-2-I}.} Let $\eta>0$ be a small parameter and $A$ be a bounded self-adjoint operator on $L^{2}(\mathbb R^{2})$. Consider the perturbed Hamiltonian in the group of $N_{1}$ particles
\begin{align}
H_{N,\eta} = &  \sum_{i=1}^{N_{1}}\big(-\Delta_{x_{i}}+V_{1}(x_{i})+\eta A_{x_{i}}\big)-\frac{1}{N_{1}-1}\sum_{1 \leq i < j \leq N_{1}}w_{N}^{(1)}(x_{i}-x_{j}) \nonumber \\
& + \sum_{r=1}^{N_{2}}\big(-\Delta_{y_{r}}+V_{2}(y_{r})\big)-\frac{1}{N_{2}-1}\sum_{1 \leq r < s \leq N_{2}}w_{N}^{(2)}(y_{r}-y_{s}) \nonumber\\
& - \frac{1}{N}\sum_{i=1}^{N_{1}}\sum_{r=1}^{N_{2}}w_{N}^{(12)}(x_{i}-y_{r}). \label{hamiltonian:perturbed-1}
\end{align}
with the ground state energy per particle denoted by $E_{\eta}^{\rm Q}$ hereafter. The associated (modified) Hartree functional is
$$
\mathcal{E}_{N,\eta}^{\rm mH} (u_{1},u_{2}) = \mathcal{E}_{N}^{\rm mH} (u_{1},u_{2}) + c_{1}\eta\langle u_{1},Au_{1}\rangle,
$$
with the corresponding (modified) Hartree energy $E_{N,\eta}^{\rm mH}$. Note that $\mathcal{E}_{N}^{\rm mH} = \mathcal{E}_{N,0}^{\rm mH}$ and $E_{N}^{\rm mH} = E_{N,0}^{\rm mH}$. From the arguments in the proof of Theorem \ref{thm:quantum-hartree-energy-general}, we have
\begin{equation}\label{estimate:energy-quantum-hartree-operator}
E_{\eta}^{\rm Q} \geq \inf_{\substack{\gamma_{1}=\gamma_{1}^{*}\geq 0, \gamma_{2}=\gamma_{2}^{*}\geq 0\\ \tr\gamma_{1} = 1 = \tr\gamma_{2}}} \tilde{\mathcal{E}}_{N,\eta}^{\rm H}(\gamma_{1},\gamma_{2}) - CN^{2\beta-1},
\end{equation}
where
\begin{align}
\tilde{\mathcal{E}}_{\eta}^{\rm H}(\gamma_{1},\gamma_{2}) & = c_{1}\left[\tr(-\Delta+V_{1}+\eta A)\gamma_{1} - \frac{a_{1}}{2}\int_{\mathbb R^2}\rho_{\gamma_{1}}(x)(w_{N}^{(1)}\star\rho_{\gamma_{1}})(x){\rm d}x\right] \nonumber \\
& \quad  + c_{2}\left[\tr(-\Delta+V_{2})\gamma_{2} - \frac{a_{2}}{2}\int_{\mathbb R^2}\rho_{\gamma_{2}}(x)(w_{N}^{(2)}\star\rho_{\gamma_{2}})(x){\rm d}x\right] \nonumber \\
& \quad  - c_{1}c_{2}a_{12}\int_{\mathbb R^{2}}\rho_{\gamma_{1}}(x)(w_{N}^{(12)}\star\rho_{\gamma_{2}})(x){\rm d}x, \label{functional:hartree-operator}
\end{align}
Here, $\rho_{\gamma}(x) := \gamma(x, x)$ for every $\gamma = \gamma^{*}\geq 0$ with $\tr\gamma = 1$. Now, we remark that the infimum on the right hand side of \eqref{estimate:energy-quantum-hartree-operator} is bounded from below by the (modified) Hartree energy $E_{N,\eta}^{\rm mH}$. To see this, we write $\gamma_{1}$ and $\gamma_{2}$ in terms of spectral representations
$$
\gamma_{\sigma} = \sum_{j}n_{j}^{(\sigma)}|u_{j}^{(\sigma)}\rangle\langle u_{j}^{(\sigma)}|, \quad \sigma\in\{1,2\},
$$
with $0\leq n^{(\sigma)}_{j} \leq 1$ and $\sum_{i}n^{(1)}_{i} = 1 = \sum_{j}n^{(2)}_{j}$. In other words, $\gamma_{\sigma}$ are convex combinations of the projections onto the orthogonal eigenfunctions $\{u_{j}^{(\sigma)}\}_{j}$. Furthermore, one can easily check that
\begin{equation}\label{density:hartree-operator}
\rho_{\gamma_{\sigma}} = \sum_{j}n_{j}^{(\sigma)}|u_{j}|^{2} \leq \sqrt{\sum_{j}n_{j}^{(\sigma)}|u_{j}|^{4}}.
\end{equation}
Fixing $\gamma_{2}$, applying \eqref{ineq:interaction} to $\kappa=1$, $i=1=j$ with $u_{1} = \sqrt{\rho_{\gamma_{1}}}$ and using \eqref{density:hartree-operator} for $\rho_{\gamma_{1}}$, we obtain
\begin{align}
\tilde{\mathcal{E}}_{\eta}^{\rm H}(\gamma_{1},\gamma_{2}) & \geq c_{1}\sum_{j}n_{j}^{(1)}\left[\langle u_{j}^{(1)},-\Delta+V_{1}+\eta A - c_{2}a_{12}w_{N}^{(12)}\star\rho_{\gamma_{2}},u_{j}^{(1)}\rangle - \frac{a_{1}}{2}\int_{\mathbb R^2}|u_{j}^{(1)}(x)|^{4}{\rm d}x\right] \nonumber \\
& \quad  + c_{2}\left[\tr(-\Delta+V_{2})\gamma_{2} - \frac{a_{2}}{2}\int_{\mathbb R^2}\rho_{\gamma_{2}}(x)(w_{N}^{(1)}\star\rho_{\gamma_{2}})(x){\rm d}x\right] \nonumber \\
& \geq c_{1}\inf_{\substack{u\in H^{1}(\mathbb R^2)\\ \|u\|_{L^{2}}=1}}\left\{\langle u,-\Delta+V_{1}+\eta A - c_{2}a_{12}w_{N}^{(12)}\star\rho_{\gamma_{2}},u\rangle - \frac{a_{1}}{2}\int_{\mathbb R^2}|u(x)|^{4}{\rm d}x\right\} \label{functional:modified-nls-one-component-1} \\
& \quad  + c_{2}\left[\tr(-\Delta+V_{2})\gamma_{2} - \frac{a_{2}}{2}\int_{\mathbb R^2}\rho_{\gamma_{2}}(x)(w_{N}^{(1)}\star\rho_{\gamma_{2}})(x){\rm d}x\right], \nonumber
\end{align}
where we have used $\sum_{j}n^{(1)}_{j} = 1$. Let $u_{N}^{(1)}\in H^{1}(\mathbb R^2)$ be a (fixed) ground state for the variational problem \eqref{functional:modified-nls-one-component-1}, which exists under the condition $0<a_{1}<a_{*}$. Applying \eqref{ineq:interaction} to $\kappa=1$, $i=2=j$ with $u_{2} = \sqrt{\rho_{\gamma_{2}}}$ and using \eqref{density:hartree-operator} for $\rho_{\gamma_{2}}$, we obtain
\begin{align}
\tilde{\mathcal{E}}_{\eta}^{\rm H}(\gamma_{1},\gamma_{2}) & \geq c_{1}\left[\langle u_{N}^{(1)},-\Delta+V_{1}+\eta A,u_{N}^{(1)}\rangle - \frac{a_{1}}{2}\int_{\mathbb R^2}|u_{N}^{(1)}(x)|^{4}{\rm d}x\right] \nonumber \\
& \quad + c_{2}\sum_{j}n_{j}^{(2)}\left[\langle u_{j}^{(2)},-\Delta+V_{2} - c_{1}a_{12}w_{N}^{(12)}\star |u_{N}^{(1)}|^{2},u_{j}^{(2)}\rangle - \frac{a_{2}}{2}\int_{\mathbb R^2}|u_{j}^{(2)}(x)|^{4}{\rm d}x\right] \nonumber \\
& \geq c_{1}\left[\langle u_{N}^{(1)},-\Delta+V_{1}+\eta A,u_{N}^{(1)}\rangle - \frac{a_{1}}{2}\int_{\mathbb R^2}|u_{N}^{(1)}(x)|^{4}{\rm d}x\right] \nonumber \\
& \quad + c_{2}\inf_{\substack{u\in H^{1}(\mathbb R^2)\\ \|u\|_{L^{2}}=1}}\left\{\langle u,-\Delta+V_{2} - c_{1}a_{12}w_{N}^{(12)}\star |u_{N}^{(1)}|^{2},u\rangle - \frac{a_{2}}{2}\int_{\mathbb R^2}|u(x)|^{4}{\rm d}x\right\}, \label{functional:modified-nls-one-component-2}
\end{align}
where we have used $\sum_{j}n^{(2)}_{j} = 1$. Let $u_{N}^{(2)}\in H^{1}(\mathbb R^2)$ be a ground state for the variational problem \eqref{functional:modified-nls-one-component-2}, which exists under the condition $0<a_{2}<a_{*}$. Then, it follows that
\begin{equation}\label{estimate:energy-quantum-hartree-operator-modified}
\tilde{\mathcal{E}}_{\eta}^{\rm H}(\gamma_{1},\gamma_{2}) \geq 	\mathcal{E}_{N,\eta}^{\rm mH} (u_{N}^{(1)},u_{N}^{(2)}) \geq E_{N,\eta}^{\rm mH},
\end{equation}
which holds for any $\gamma_{1}=\gamma_{1}^{*}\geq 0$ and $\gamma_{2}=\gamma_{2}^{*}\geq 0$ such that $\tr\gamma_{1} = 1 = \tr\gamma_{2}$.  Finally, we conclude from \eqref{estimate:energy-quantum-hartree-operator-modified} and \eqref{estimate:energy-quantum-hartree-operator} that  we must have
\begin{equation}\label{lower:quantum-modified-hartree}
E_{\eta}^{\rm Q} \geq E_{N,\eta}^{\rm  mH} - CN^{2\beta-1}.
\end{equation}

Let $(u_{1,N,\eta},u_{2,N,\eta})$ and $\Psi_{N}$ be ground states for $E_{N,\eta}^{\rm mH}$ and $H_{N}=H_{N,0}$, respectively. We have
\begin{align}
\eta c_{1}\tr(A\gamma_{\Psi_{N}}^{(1,0)}) & = N^{-1}\langle\Psi_{N},H_{N,\eta}\Psi_{N}\rangle - N^{-1}\langle\Psi_{N},H_{N}\Psi_{N}\rangle  \nonumber\\
& \geq E_{\eta}^{\rm Q} - E_{N}^{\rm Q} \nonumber\\
& \geq E_{N,\eta}^{\rm mH} - E_{N}^{\rm mH} + \mathcal{O}(N^{\frac{3}{p_{0}+2}\gamma-\beta}) + \mathcal{O}(N^{2\beta-1}) \nonumber\\
& \geq \eta c_{1}\langle u_{1,N,\eta},Au_{1,N,\eta}\rangle + \mathcal{O}(N^{\frac{3}{p_{0}+2}\gamma-\beta}) + \mathcal{O}(N^{2\beta-1}), \label{FH}
\end{align}
where $p_{0}=\min\{p_{1},p_{2}\}$. In the above, the first and the last inequalities are the variational principles. The second estimate used \eqref{lower:quantum-modified-hartree} and \eqref{upper:quantum-modified-hartree}. Now, under the assumption that
$$
0<\gamma<\min\left\{\frac{p_{0}+2}{p_{0}+3}\beta,\frac{p_{0}+2}{p_{0}}(1-2\beta)\right\},
$$
one can pick $\eta=\eta_{N}=o\big((a_{*}-a_{N})^{\frac{p_{0}}{p_{0}+2}}\big)$ such that
\begin{equation} \label{approximate-ground-state-I}
\lim_{N\to\infty}\eta^{-1}(N^{\frac{3}{p_{0}+2}\gamma-\beta} + N^{2\beta-1}) = 0.
\end{equation}
Then, it follows from \eqref{FH} and repeating the argument with $A$ replaced by $-A$ that
\begin{equation} \label{conv:ground-state-1-2-I}
\langle u_{1,N,\eta},Au_{1,N,\eta}\rangle + o(1)_{N\to\infty} \leq \tr(A\gamma_{\Psi_{N}}^{(1,0)}) \leq \langle u_{1,N,-\eta},Au_{1,N,-\eta}\rangle + o(1)_{N\to\infty}.
\end{equation}
On the other hand, since $(u_{1,N,\eta},u_{2,N,\eta})$ is a ground state for $E_{N,\eta}^{\rm mH}$ (recall that $\mathcal{E}_{N}^{\rm mH} = \mathcal{E}_{N,0}^{\rm mH}$ and $E_{N}^{\rm mH} = E_{N,0}^{\rm mH}$), we have, with the choice of $\eta$ in \eqref{approximate-ground-state-I},
$$
E_{N}^{\rm mH} \leq \mathcal{E}_{N}^{\rm mH}(u_{1,N,\eta},u_{2,N,\eta}) \leq E_{N,\eta}^{\rm mH} + \eta\|A\| \leq \mathcal{E}_{N,\eta}^{\rm mH}(u_{1,N,0},u_{2,N,0}) + \eta\|A\| \leq E_{N}^{\rm mH} + 2\eta\|A\|.
$$
It follows from the above that $(u_{1,N,\eta},u_{2,N,\eta})$ and $(u_{1,N,-\eta},u_{2,N,-\eta})$ are sequences of quasi-ground states for $E_{N,\eta}^{\rm mH}$. We may apply Theorem \ref{thm:blowup-modified-hartree-approximate-1-2-I} together with \eqref{conv:ground-state-1-2-I} to  get the trace-class weak-$\star$ convergence of $\gamma_{\Psi_{N}}^{(1,0)}$ to $|Q_{N}\rangle \langle Q_{N}|$, where $\ell_{N}$ is defined as in Theorem \ref{thm:blow-up-bec-1-2-I} and $Q_{N} = \ell_{N}Q_{0}(\ell_{N}\cdot)$ with $Q_{0}$ given by \eqref{GN:normalized}. Since no mass is lost in the limit, the convergence must hold in trace-class norm. 

Now, we consider the perturbed Hamiltonian in the group of $N_{2}$ particles
\begin{align}
H_{N,\eta} = &  \sum_{i=1}^{N_{1}}\big(-\Delta_{x_{i}}+V_{1}(x_{i})\big)-\frac{1}{N_{1}-1}\sum_{1 \leq i < j \leq N_{1}}w_{N}^{(1)}(x_{i}-x_{j}) \nonumber \\
& + \sum_{r=1}^{N_{2}}\big(-\Delta_{y_{r}}+V_{2}(y_{r})+\eta A_{y_{r}}\big)-\frac{1}{N_{2}-1}\sum_{1 \leq r < s \leq N_{2}}w_{N}^{(2)}(y_{r}-y_{s}) \nonumber \\
& - \frac{1}{N}\sum_{i=1}^{N_{1}}\sum_{r=1}^{N_{2}}w_{N}^{(12)}(x_{i}-y_{r}).\label{hamiltonian:perturbed-2}
\end{align}
At this stage, by repeating the above argument, we also obtain the convergence of $\gamma_{\Psi_{N}}^{(0,1)}$ to $|Q_{N}\rangle \langle Q_{N}|$ in trace-class norm. Equivalently, both $\gamma_{\Phi_{N}}^{(1,0)}$ and $\gamma_{\Phi_{N}}^{(0,1)}$ converge in trace-class norm to a rank-one operator $|Q_{0}\rangle \langle Q_{0}|$, where $\Phi_{N} = \ell_{N}^{-N}\Psi_{N}(\ell_{N}^{-1}\cdot)$. It is well known that this implies the convergence of generic higher order density matrices to tensor powers of the limiting operator (see, e.g., the discussion following \cite[Section 3]{MicOlg-17}).

\subsection{Proof of Theorem \ref{thm:blow-up-bec-II}}\label{subsec:proof-II}

The purpose of this subsection is to prove Theorem \ref{thm:blow-up-bec-II}, which gives the blow-up profile of the many-body system \eqref{hamiltonian} when the interaction strength of intra-species among particles in each component tends to a critical number. We first revisit the blow-up phenomenon for the NLS minimization problem \eqref{functional:nls-two-component}. In the case $a_{12}<0$, the existence of the NLS ground states follows the standard direct method in the calculus of variations. In \cite[Theorem 1.1]{GuoZenZho-18}, the authors proved that if $0<a_{1},a_{2}<a_{*}$, then $E^{\rm NLS}\geq 0$ and it has at least one ground state. On the other hand, $E^{\rm NLS}=-\infty$ when either $a_{1}>a_{*}$ or $a_{2}>a_{*}$. This is somehow similar to the one-component setting (see \cite{GuoSei-14}).

The next result concerns the limit behavior of NLS energy and its ground states. For the system of attractive intra-species interactions and repulsive inter-species interactions, we are not able to determine the accurate blow-up rate of ground states when $z_{1} = z_{2}$, due to the absence of a refined energy estimate. When $z_{1} \ne z_{2}$, the decays property \eqref{decay:exponential} allows us to control the cross-term in \eqref{functional:nls-two-component}. In this case, the authors in \cite{GuoZenZho-18} proved that for any fixed $a_{12}<0$, we have 
\begin{equation}\label{behavior:nls-energy-II}
c_{1}E_{1}^{\rm NLS} + c_{2}E_{2}^{\rm NLS} \leq E^{\rm NLS} \leq c_{1}E_{1}^{\rm NLS} + c_{2}E_{2}^{\rm NLS} + C\left(e^{-\mu_{0}(a_{*}-a_{1})^{-\frac{1}{p_{2}+2}}} + e^{-\mu_{0}(a_{*}-a_{2})^{-\frac{1}{p_{2}+2}}}\right).
\end{equation}
Here, $E_{i}^{\rm NLS}$ is defined in \eqref{energy:nls-one-component}  and  $\mu_{0} = \mu|z_{1}-z_{2}|>0$ with $\mu>0$ given by \eqref{decay:exponential}. The limiting profile of the NLS \emph{exact} ground states follows from the energy estimate \eqref{behavior:nls-energy-II} and the asymptotic behavior of the NLS energy in the one-component setting, which is given by \eqref{lambda:II}. For the detailed analysis, we refer the reader to \cite[Theorem 1.2]{GuoZenZho-18}.

In this subsection, we only study the blow-up profile of the many-body system \eqref{hamiltonian} in the case $z_{1} \ne z_{2}$. The asymptotic formula of the quantum energy which is given in Theorem \ref{thm:blow-up-bec-II} follows from that of the NLS energy and the estimate
\begin{equation}\label{asymptotic:quantum-energy-II}
c_{1}E_{1}^{\rm NLS} + c_{2}E_{2}^{\rm NLS} - CN^{2\beta-1} \leq E_{N}^{\rm H} - CN^{2\beta-1} \leq E_{N}^{\rm Q} \leq E_{N}^{\rm H} \leq c_{1}E_{1}^{\rm NLS} + c_{2}E_{2}^{\rm NLS} + \mathcal{R}_{5},
\end{equation}
where we abbreviated by $\mathcal{R}_{5}$ the error terms, with $\ell_{i,N}=\Lambda_{i} (a_{*}-a_{i,N})^{-\frac{1}{p_{i}+2}}$ for $i\in\{1,2\}$,
\begin{equation}\label{err:energy}
\mathcal{R}_{5} := CN^{-\beta}(\ell_{1,N}^{3} + \ell_{2,N}^{3}) + C(e^{-\mu_{0}\ell_{1,N}} + e^{-\mu_{0}\ell_{2,N}}).
\end{equation}
In \eqref{asymptotic:quantum-energy-II}, the first inequality follows from \eqref{ineq:interaction} and the non-negativity of the inter-species interactions. The second and the third inequalities used \eqref{lower:quantum-hartree-energy}. Finally, the last inequality follows from \eqref{upper:hartree-nls-energy} and \eqref{behavior:nls-energy-II}.

Similarly to what was done in Section \ref{subsec:proof-I}, we need to extend the blow-up result in \cite{GuoZenZho-18} to the (modified) Hartree ground states in order to obtain the convergence of the many-body ground states. Fortunately, the inter-species interactions are repulsive and we can ignore it in the energy estimate. Therefore, an extension to the blow-up behavior of the NLS \emph{approximate} ground states is sufficient for our purpose. We have the following.

\begin{theorem}\label{lem:blowup-nls-approximate-II}
	Assume that $V_{i}$, for $i\in\{1,2\}$, are defined as in \eqref{potentials:external} with $z_{i}\in\mathbb R^2$. Let $a_{i,N}\nearrow a_{*}$ as $N\to\infty$. Let $u_{i,N} \in H^{1}(\mathbb{R}^{2})$ be a sequence of positive functions such that $\|u_{i,N}\|_{L^{2}}=1$ and 
	\begin{equation}\label{behavior:approximate-nls-energy-II}
	\mathcal{E}_{i}^{\rm NLS}(u_{i,N}) = E_{i}^{\rm NLS} + o(E_{i}^{\rm NLS})_{N\to\infty} = (a_{*}-a_{i,N})^{\frac{p_{i}}{p_{i}+2}}\left(\frac{p_{i}+2}{p_{i}} \cdot \frac{\Lambda_{i}^{2}}{a_{*}}+o(1)_{N\to\infty}\right)
	\end{equation}
	where $\mathcal{E}_{i}^{\rm NLS}$ are defined in \eqref{functional:nls-one-component} and $\Lambda_{i}$ are given by \eqref{lambda:II}. Then, we have
	\begin{equation}\label{conv:approximate-nls-one-component}
	\lim_{N\to\infty} \ell_{i,N}^{-1} u_{i,N}(\ell_{i,N}^{-1}\cdot+z_{i}) = Q_{0}
	\end{equation}
	strongly in $L^{2}(\mathbb R^2)$, where $\ell_{i,N}=\Lambda_{i} (a_{*}-a_{i,N})^{-\frac{1}{p_{i}+2}}$ and $Q_{0}$ is given by \eqref{GN:normalized}.
\end{theorem}

\begin{proof}
	Proof of Theorem \ref{lem:blowup-nls-approximate-II} can be found in \cite[Theorem 2.1]{LewNamRou-17-proc}.
\end{proof}

Now, we are in the position to give the proof of Theorem \ref{thm:blow-up-bec-II} which closely follows that of Theorem \ref{thm:blow-up-bec-1-2-I}. For the reader's convenience, we will mention only the differences in the following.

\begin{proof}[Proof of Theorem \ref{thm:blow-up-bec-II}.]
	Let $\eta>0$ be a small parameter and $A$ be a bounded self-adjoint operator on $L^{2}(\mathbb R^{2})$. Consider the perturbed Hamiltonian in the group of $N_{1}$ particles as in \eqref{hamiltonian:perturbed-1} with the ground state energy per particle denoted by $E_{\eta}^{\rm Q}$. The associated NLS functional is
	$$
	\mathcal{E}_{\eta}^{\rm NLS} (u_{1},u_{2}) = c_{1}\mathcal{E}_{1,\eta}^{\rm NLS}(u_{1}) + c_{2}\mathcal{E}_{2}^{\rm NLS}(u_{2}) - c_{1}c_{2}a_{12}\int_{\mathbb R^{2}}|u_{1}(x)|^{2}|u_{2}(x)|^{2}{\rm d}x,
	$$
	where we have introduced the NLS energy functional
	$$
	\mathcal{E}_{1,\eta}^{\rm NLS}(u) = \mathcal{E}_{1}^{\rm NLS}(u) + \eta\langle u,Au \rangle
	$$
	with the corresponding NLS energy $E_{1,\eta}^{\rm NLS}$. Note that $\mathcal{E}_{1,0}^{\rm NLS} = \mathcal{E}_{1}^{\rm NLS}$ and $E_{1,0}^{\rm NLS} = E_{1}^{\rm NLS}$. Let $u_{1,N,\eta}$ be a ground state for $E_{1,\eta}^{\rm NLS}$. Let $\Psi_{N}$ be a ground state for $H_{N}=H_{N,0}$. Then, it follows from \eqref{asymptotic:quantum-energy-II} that
	\begin{align*}
	\eta c_{1}\tr(A\gamma_{\Psi_{N}}^{(1,0)}) & = N^{-1}\langle\Psi_{N},H_{N,\eta}\Psi_{N}\rangle - N^{-1}\langle\Psi_{N},H_{N}\Psi_{N}\rangle \\
	& \geq E_{\eta}^{\rm Q} - E_{N}^{\rm Q} \\
	& \geq c_{1}E_{1,\eta}^{\rm NLS} + c_{2}E_{2}^{\rm NLS} - c_{1}E_{1}^{\rm NLS} - c_{2}E_{2}^{\rm NLS} - CN^{2\beta-1} - \mathcal{R}_{5}\\
	& \geq \eta c_{1}\langle u_{1,N,\eta},Au_{1,N,\eta}\rangle - CN^{2\beta-1} - \mathcal{R}_{5}.
	\end{align*}
	Here, $\mathcal{R}_{5}$ is defined as in \eqref{err:energy}. Under the assumption that $a_{*}-a_{1,N} = N^{-\gamma_{1}}$ and $a_{*}-a_{2,N} = N^{-\gamma_{2}}$ with
	$$
	\frac{\gamma_{1}}{\gamma_{2}} = \frac{p_{1}+2}{p_{2}+2} \quad \text{and} \quad 0<\gamma_{i}<\min\left\{\frac{p_{i}+2}{p_{i}+3}\beta,\frac{p_{i}+2}{p_{i}}(1-2\beta)\right\},\quad i\in\{1,2\},
	$$
	one can pick $\eta=\eta_{1,N}=o\big((a_{*}-a_{1,N})^{\frac{p_{1}}{p_{1}+2}}\big)$ such that
	\begin{equation} \label{approximate-ground-state-II}
	\lim_{N\to\infty}\eta^{-1}(CN^{2\beta-1} - \mathcal{R}_{5}) = 0.
	\end{equation}
	The rest of the proof is the same as for Theorem \ref{thm:blow-up-bec-II}. The convergence of the many-body ground states follows from that of the one-component NLS \emph{approximate} ground states $u_{1,N,\eta}$  and $u_{2,N,\eta}$, which was given in Theorem \ref{lem:blowup-nls-approximate-II}. We omit the details for brevity.
\end{proof}

\section*{Acknowledgements}

The author is indebted to the referee for many useful suggestions which improved significantly the presentation of the paper. Also, he is very grateful to T. K\"onig for his proofreading of the manuscript. He cordially thanks A. Triay and X. Zeng for some  helpful discussions. This work was funded by the Deutsche Forschungsgemeinschaft (DFG, German Research Foundation) under Germany's Excellence Strategy EXC-2111-390814868.

\end{document}